\documentclass[conference,a4paper]{IEEEtran}

\usepackage[utf8]{inputenc} 
\usepackage[T1]{fontenc}
\usepackage{url}              % provides \url{...}
\usepackage{cite}             % improves presentation of citations
\usepackage{xcolor}
\usepackage[cmex10]{amsmath}  % Use the [cmex10] option to ensure complicance
\usepackage{mleftright}       % fix to wrong spacing of \left-,
\mleftright                   % \middle- \right-commands 

\usepackage{graphicx}         % provides \includegraphics{...} to
\usepackage{booktabs}         % fixes poor spacing in tables and
\usepackage{amssymb, amsthm}
\usepackage[colorinlistoftodos,backgroundcolor=lightgray!50,linecolor=black,textsize=tiny]{todonotes}

\DeclareMathOperator*{\argmin}{arg\,min}

\newcommand\myeq{\mkern1.2mu{=}\mkern1.2mu}

\newcommand\myleq{\mkern1.2mu{<}\mkern1.2mu}

\newcommand\myminus{\mkern1.2mu{-}\mkern1.2mu}

\newcommand{\boldY}{\mathbf{Y}}     % Y of P
\newcommand{\boldZ}{\mathbf{Z}}     % Z of P
\newcommand{\boldYQ}{\mathbf{Y}}   % Y' of Q (vector). Note: reverted back to boldY
\newcommand{\boldZQ}{\mathbf{Z}}   % Z' of Q (vector)
\newcommand{\YQ}{Y}   % Y' of Q (scalar)
\newcommand{\ZQ}{Z}   % Z' of Q (scalar)
\newcommand{\boldK}{\mathbf{K}}
\newcommand{\Po}{\text{Poisson}}
\newcommand{\R}{\mathbb{R}}
\newcommand{\Rplus}{\mathbb{R}^+}
\newcommand{\N}{\mathbb{N}}
\newcommand{\Nnot}{\mathbb{N}_0}
\newcommand{\Multi}{\text{Multinomial}}
\newcommand{\st}{\mkern1.2mu{:}\mkern1.2mu}
\newcommand{\mzsum}{\sum_{i=1}^{k_2-1} z_i}

\newcommand{\ky}{k_1}
\newcommand{\kz}{k_2}
\newcommand{\sumoverz}{z_s}
\newcommand{\sumoverq}{q_s}

\newcommand{\sizemultb}{s_z}
\newcommand{\msg}{M}
\newcommand{\msgsmall}{m}
\newcommand{\boldlambda}{\boldsymbol{\Lambda}}
\newcommand{\idxset}{\mathcal{I}}

\newcommand{\idxsetc}{\mathcal{I}^0}
\newcommand{\sumset}{C}
\newcommand{\boldy}{\mathbf{y}}
\newcommand{\boldz}{\mathbf{z}}
\newtheorem{theorem}{Theorem}
\newtheorem{lemma}{Lemma}
\newtheorem{corollary}{Corollary}

\begin{document}

\title{Computing Unique Information for Poisson and Multinomial Systems} 

%%%%%%
% \author{%
%   \IEEEauthorblockN{Anonymous Authors}
%   \IEEEauthorblockA{%
%     Please do NOT provide authors' names and affiliations\\
%     in the paper submitted for review, but keep this placeholder.\\
%     ISIT23 follows a \textbf{double-blind reviewing policy}.}
% }

%%%%%% Please only add the author names and affiliations for the FINAL
%%%%%% version of the paper, but NOT for the paper submitted for review!
%
%%%%%
%%%%% Single author, or several authors with same affiliation:
% \author{%
%   \IEEEauthorblockN{Stefan M.~Moser}
%   \IEEEauthorblockA{ETH Zürich\\
%                     8092 Zürich, Switzerland\\
%                     moser@isi.ee.ethz.ch}
%                   }

\author{
  \IEEEauthorblockN{Chaitanya Goswami, Amanda Merkley, and Pulkit Grover}
  \IEEEauthorblockA{Department of Electrical and Computer Engineering,\\
                    Carnegie Mellon University, Pittsburgh, PA\\
                    Email: \{cgoswami, amerkley, pgrover\}@andrew.cmu.edu}
        }

\maketitle

%%%%%
%% Abstract: 
%% If your paper is eligible for the student paper award, please add
%% the comment "THIS PAPER IS ELIGIBLE FOR THE STUDENT PAPER
%% AWARD." as a first line in the abstract. 
%% For the final version of the accepted paper, please do not forget
%% to remove this comment!
%%
\begin{abstract}
Bivariate Partial Information Decomposition (PID) describes how the mutual information between a random variable $\msg$ and two random variables $Y$ and $Z$ is decomposed into unique, redundant, and synergistic terms. Recently, PID has shown promise as an emerging tool to understand biological systems and biases in machine learning. However, computing PID is a challenging problem as it typically involves optimizing over distributions. In this work, we study the problem of computing PID in two systems: the Poisson system inspired by the ``ideal Poisson channel'' and the multinomial system inspired by multinomial thinning, for a scalar $\msg$. We provide sufficient conditions for both systems under which closed-form expressions for many operationally-motivated PID can be obtained, thereby allowing us to easily compute PID for these systems. Our proof consists of showing that one of the unique information terms is zero, which allows the remaining unique, redundant, and synergistic terms to be easily computed using only the marginal and the joint mutual information.

\end{abstract}

%%%%%%%%%%%%%%%%%%%%%%%%%%%%%% SECTION %%%%%%%%%%%%%%%%%%%%%%%%%%%%%%
\section{Introduction}\label{sec:intro}
In recent years, the problem of bivariate Partial Information Decomposition (PID)\footnote{Throughout this work we refer to bivariate PID as PID.}, i.e. decomposing the total information that random variables $Y$ and $Z$ contain about a random variable $\msg$ into different components, has received much attention~\cite{williams2010nonnegative,harder2013bivariate,griffith2014quantifying,bertschinger2014quantifying,gutknecht2021bits,banerjee2018unique,kay2018exact,finn2018pointwise,niu2019measure,rosas2020operational,gurushankar2022extracting,venkatesh2022partial,barrett2015exploration}. Formally, PID refers to a framework for describing how the total mutual information $I(\msg;Y,Z)$ can be partitioned into four components: (i) $UI(\msg; Y \backslash Z)$, the unique information $Y$ has about $\msg$ that is not in $Z$; (ii) $UI(\msg; Z \backslash Y)$, the unique information $Z$ has about $\msg$ that is not in $Y$; (iii) $RI(\msg; Y, Z)$, the redundant information shared between $Y$ and $Z$ about $\msg$; (iv) $SI(M; Y, Z)$, the synergistic information about $\msg$ attained through $Y$ and $Z$ jointly. PID has been used to examine biological systems~\cite{colenbier2020disambiguating,boonstra2019information, krohova2019multiscale,pica2017quantifying,gat1998synergy,venkatesh2020information}, quantify biases in machine learning~\cite{dutta2020information}, and analyze financial data~\cite{scagliarini2020synergistic}.  

Substantial work in the PID literature has been focused on defining the exact framework used in PID~\cite{williams2010nonnegative,harder2013bivariate,griffith2014quantifying,bertschinger2014quantifying,banerjee2018unique,kay2018exact,finn2018pointwise,niu2019measure,rosas2020operational,gurushankar2022extracting, gutknecht2021bits}, resulting in many competing measures. A PID measure of particular interest is the one proposed in~\cite{bertschinger2014quantifying}, for which the authors provide a decision-theoretic operational interpretation, further motivated in Sec. \ref{sec:PID_background}.  
However, limited work has investigated efficient computation of PID, with all the works focusing on the case where $\msg$, $Y$, and $Z$ are either jointly Gaussian~\cite{venkatesh2022partial, barrett2015exploration, gurushankar2022extracting}, or are discrete random variables~\cite{makkeh2017bivariate,banerjee2018computing,makkeh2018broja}. This problem is exacerbated since operationally well-motivated PID measures, e.g.~\cite{bertschinger2014quantifying}, are hard to compute, involving an optimization over distributions~\cite{venkatesh2022partial}. 

A notable work in computing the PID terms is~\cite{barrett2015exploration}, where the authors show that for the special case of a scalar $M$, and $M,Y,Z$ jointly Gaussian, several measures of PID (including~\cite{bertschinger2014quantifying,williams2010nonnegative,harder2013bivariate,griffith2014quantifying}) can be reduced to a much simpler form since one of the unique information (UI) terms is zero. All remaining PID terms are easily derived using certain desirable properties satisfied in many PID frameworks (see Sec~\ref{sec:PID_background}). This result makes computing the PID terms for Gaussian systems having scalar $\msg$ extremely easy, as it avoids the issue of optimizing over distributions by providing closed form expressions of mutual information terms.

In this work, we provide an analogous result for two systems: (i) \emph{The Poisson System}: $\msg$ is a scalar non-negative random variable, and $Y$ and $Z$ conditioned on $\msg$ follow a multivariate Poisson distribution, with the dependence of $Y$ and $Z$ on $\msg$ inspired by a well-known Poisson channel~\cite{atar2012mutual} (see Sec.~\ref{sec:problem_setup}), and (ii) \emph{The Multinomial System}: $\msg$ is a scalar random variable with support over positive integers, and $Y$ and $Z$ are multinomially thinned~\cite{chandramohan1985bernoulli} version of $\msg$ (see Sec.~\ref{sec:multinomial}). 
%\question{The multivariate Poisson and multinomial are canonical distributions for describing general count and categorical data} 
We derive \textit{sufficient} conditions for both systems under which at least one of the UI terms is zero. Hence, as in the Gaussian case \cite{barrett2015exploration}, the PID in the Poisson and multinomial systems is reduced to closed form expressions.

Our proof technique relies on the fact that under the proposed conditions, we can construct the Markov chain of the form $\msg\rightarrow \YQ \rightarrow \ZQ$, such that its marginals $Q(\msg, \YQ)$ and $Q(\msg, \ZQ)$ are the same as the marginals of the original system: $P(M,Y)$ and $P(\msg,Z)$.\footnote{{Technically, the Markov chain should be $M\rightarrow Y'\rightarrow Z'$, where $Y'$ and $Z'$ are random variables such that $Q(\msg, Y'\myeq y)=P(\msg,Y\myeq y)$ and $Q(\msg, Z'=z)=P(\msg,Z=z)$. To maintain consistency with the notation of~\cite{bertschinger2014quantifying}, we abuse notation and use $M\rightarrow Y\rightarrow Z$ instead of $M\rightarrow Y'\rightarrow Z'$, where $M\rightarrow Y\rightarrow Z$ is to be interpreted as an alternative joint distribution on $\msg,Y$ and $Z$ having a Markov structure.}} Following Theorem 2 of~\cite{venkatesh2022partial}, if the aforementioned Markov chain can be constructed, then for many PID measures, one of the UI terms must be zero\footnote{Authors of~\cite{venkatesh2022partial} refer to the existence of the Markov chain $\msg\rightarrow\boldY\rightarrow\boldZ$ having the same marginals $P(\msg,\boldY)$ and $P(\msg,\boldZ)$ as $P(\boldY|\msg)$ being ``Blackwell sufficient''~\cite{blackwell1953equivalent} for $P(\boldZ|\msg)$.}. For the specific PID measure proposed in~\cite{bertschinger2014quantifying}, Lemma~\ref{lemma:poisson_indep:1} explicitly shows how the existence of the aforementioned Markov chain implies that one of the UI terms is zero. We refer the readers to the Section III-A, and the Appendix B Part C of~\cite{venkatesh2022partial} for an explicit characterization of all the PID measures to which our result applies.

Our particular interest in the Poisson and multinomial systems stems from applications in neuroscience, where Poisson distribution is widely used to model neural spikes, and binomial thinning (a special case of multinomial thinning) is a widely used operator in modelling neural data~\cite{dayan2005theoretical}. Our results greatly facilitate the estimation of synergy (and PID, in general) for many neural systems, a question of great interest in the neuroscience community~\cite{schneidman2003synergy, park2020deep, varadan2006computational}. The Poisson system has also been used in many other fields, such as optical communication~\cite{bar1969communication,verdu1999poisson,grandell1997mixed,dayan2005theoretical},
and multinomial systems have been used in molecular communication~\cite{yilmaz2014arrival,farsad2017capacity}.

\section{Background}\label{sec:problem_setting}
\subsection{Definitions and Notations}\label{sec:defn}
%%%%%%%%%%%%%%%%%%%%%%%%%%%%%%%%%%%%%%%%%%
%%%%%%%%%%%%%%%%%%%%%%%%%%%%%%%%%%%%%%%%%%
%%%%%%%%%% General Notations %%%%%%%%%%%%%
\textbf{Notation}: Denote the set of all natural numbers, real numbers and positive real numbers as $\N$, $\R$, and $\Rplus$, respectively. Define $\Nnot=\N\cup\{0\}$ and let $[d]=\{1,\hdots,d\}\ \forall\ d\in\N$. Define $\mathbb{A}_{i}^{d}=\{(j_1,\hdots,j_i) \st j_1<j_2<\hdots<j_i, \text{ and } j_1,\hdots,j_i\in[d]\}$, e.g. $\mathbb{A}_2^3=\{(1,2),(1,3),(2,3)\}$. We denote $\mathbf{1}$ as the vector having all elements as $1$ (the dimension of the vector $\mathbf{1}$ can be deduced by context). For brevity, the probability notations of the form $P(A|B)$, and $P(A)$ are always understood to be as $P(A=a|B=b)$, and $P(A=a)$, respectively.  

% %%%%%%%%%%%%%%%%%%%%%%%%%%%%%%%%%%%%%%%%%%
% %%%%%%%%%%%%%%%%%%%%%%%%%%%%%%%%%%%%%%%%%%
% %%%%%%%%%%%%% Definition of Poisson %%%%%%
% \textbf{Poisson Distribution}: A random variable $K$ is said to follow a Poisson distribution with rate parameter $\lambda\in\Rplus$, i.e. $K\sim\Po(\lambda)$, if it has the following probability mass function (p.m.f.):
% \begin{align}
% \Pr(K=k)=\Po(k;\lambda) = \frac{e^{-\lambda}\lambda^k}{k!},\ \forall\ k\in\Nnot.
% \end{align}

%%%%%%%%%%%%%%%%%%%%%%%%%%%%%%%%%%%%%%%%%%
%%%%%%%%%%%%%%%%%%%%%%%%%%%%%%%%%%%%%%%%%%
%%% Definition of Multivariate Poisson %%%
\textbf{Multivariate Poisson Distribution}: 
An intuitive way to define multivariate Poisson distributions is to represent each random variable in the multivariate Poisson distribution as a sum of independent Poisson random variables~\cite{fish2022interaction,krishnamoorthy1951multivariate,mahamunulu1967note,karlis2005multivariate}. To illustrate, let us construct a bivariate Poisson random vector $\boldK=\begin{bmatrix}K_1&K_2\end{bmatrix}^T$, where 
\begin{align}
    K_1 = K_1^g+K_{1,2}^g,\\ 
    K_2 = K_2^g+K_{1,2}^g,
\end{align}
and $K_1^g,K_2^g$, and $K_{1,2}^g$ are mutually independent `generator' Poisson random variables with rates $\lambda_1, \lambda_2$, and $\lambda_{12}$, respectively. Here, the dependence between $K_1$ and $K_2$ is expressed through $K_{1,2}^g$, with covariance between $K_1$ and $K_2$ being equal to $\lambda_{1,2}$~\cite{johnson1997discrete}. 

For dimensions $d>2$, the Poisson random vector $\boldK=\begingroup % keep the change local
\setlength\arraycolsep{4pt}
\begin{bmatrix}
    K_1&\hdots&K_d
\end{bmatrix}^T\endgroup$ can be similarly defined with:
\begin{align}
    K_l = \sum_{j=1}^{d}\sum_{(i_1,\hdots, i_j)\in\mathbb{B}_{jl}^{d}}K^g_{i_1,\hdots, i_j} \forall\ l\in[d],\ \text{where}\label{eq:K_l}\\
    \mathbb{B}_{jl}^{d} = \{(i_1,\hdots,i_j)\in\mathbb{A}_{j}^{d}|i_k = l\text{ for some }k\in[j]\},\nonumber
\end{align}
and $K_{i_1,\hdots, i_j}^g\sim\Po(\lambda_{i_1,\hdots, i_j})$  $\forall\ (i_1,\hdots,i_j)\in\mathbb{A}_{j}^{d}, j\in[d]$. Furthermore, the random variables $\{K_1^g,\hdots,K_{1,\hdots, d}^g\}$ are mutually independent.
\\
The $d$-dimensional random vector $\boldK$ consists of Poisson-distributed elements described by the parameter vector $\boldsymbol{\Lambda}=\begingroup 
\setlength\arraycolsep{2pt}
\begin{bmatrix}\lambda_1&\hdots&\lambda_d&\lambda_{1,2}&\hdots\lambda_{1,\hdots, d}
\end{bmatrix}^T\endgroup$.
\\
Letting $\boldK^g=[ K_1^g\dots K_d^g\  K_{1,2}^g\dots $ $K_{(d-1),d}^g\dots\ K_{1,\hdots, d}^g]^T$, 
% \begin{align*}
%     \boldK^g = \begingroup\setlength\arraycolsep{1.6pt}\begin{bmatrix}
%     K_1^g&\dots&K_d^g&K_{12}^g&\dots&K_{(d-1)d}^g&\dots&K_{d-(d'-1)\hdots d}^g
% \end{bmatrix}^T, \endgroup
% \end{align*}
we rewrite~\eqref{eq:K_l} in its vector form:
\begin{align}
    \boldK = A\boldK^g,\label{eq:poisson:dependent_defn}
\end{align}
where $A$ is a matrix of $0$'s and $1$'s indicating which generator elements from $\boldK^g$ are included in the overall variable $K_i$. We can decompose $A=\begin{bmatrix} A_1&\hdots&A_{d}\end{bmatrix}$, where $A_i$ is a $d\times \binom{d}{i}$ submatrix having no duplicate columns and each of its columns contains exactly $i$ ones and $(d-i)$ zeros~\cite{karlis2005multivariate}.
\\
An intuitive way to think about this multivariate Poisson distribution is to interpret the covariance effects in an ANOVA-like fashion~\cite{karlis2005multivariate}. For example, the ``main effect'' is given by
\begin{align}
    K_1^g, K_2^g, \dots, K_d^g, \nonumber
\end{align}
the ``two-way covariance effect'' is given by
\begin{align}
    K_{1, 2}^g, K_{1, 3}^g, \dots, K_{(d-1), d}^g, \nonumber
\end{align}
and so on until the $d$-way covariance effect:
\begin{align}
    K_{1, 2, \dots, d}^g. \nonumber
\end{align}
For a more detailed discussion, see~\cite{mahamunulu1967note,karlis2005multivariate,johnson1997discrete,krishnamoorthy1951multivariate}.
\\
Rather than specifying the full covariance structure, we can truncate dependence to $d'$-way covariance, where $d' < d$, by removing all generator variables in $\boldK^g$ associated to higher-order dependence. This is achieved by redefining $\boldK^g \myeq [ K_1^g\dots K_d^g\  K_{1,2}^g\dots $ $K_{(d-1),d}^g\dots\ K_{d-(d'-1),\hdots, d}^g]^T$ and letting the outer sum of~\eqref{eq:K_l} go to $d'$ instead of $d$. We define the multivariate Poisson distribution truncated at $d'$ covariance by $\boldK \sim \Po(d, d', \boldlambda)$.

The p.m.f. of the multivariate $\Po(d,d',\boldsymbol{\Lambda})$ distribution is as follows. Let $\mathbf{k}'\myeq\begingroup 
\setlength\arraycolsep{1.6pt}
\begin{bmatrix}
    k_{1,2}'&\hdots&k_{(d-1),d}'&\hdots&k_{d-(d'-1),\hdots, d}'
\end{bmatrix}^T\endgroup$, and $d_{\mathbf{k}'}$ be the dimension of $\mathbf{k}'$, then:
\begin{align}
    &P(\boldK\myeq\mathbf{k})\myeq e^{-\mathbf{1}^T\boldsymbol{\Lambda}}\prod_{i=1}^{d}\lambda_i^{k_i}\sum_{\mathbf{k'}\in C}\left(\prod_{(i_1,i_2)\in \mathbb{A}_2^{d}}\left(\frac{\lambda_{i_1,i_2}}{\lambda_{i_1}\lambda_{i_2}}\right)^{k_{i_1,i_2}'}\times\right.\nonumber\\
    &\hdots\times\left.\prod_{(i_1,\hdots, i_{d'})\in\mathbb{A}_{d'}^{d}}\left(\frac{\lambda_{i_1,\hdots, i_{d'}}}{\prod_{j=1}^{d'}\lambda_{i_{j}}}\right)^{k_{i_1,\hdots, i_{d'}}'}\times Q(\mathbf{k},\mathbf{k}')\right),\label{eq:poisson_multivariate:pmf}
\end{align}
 where $C=\{\mathbf{k}'\in\Nnot^{d_{\mathbf{k}'}} \st (a_i')^T\mathbf{k}'\leq k_i \ \forall\  \in [d]\}$, and
% \begin{align}
%   Q(\mathbf{k},\mathbf{k}')=\prod_{i=1}^{d}\frac{1}{(k_i-a_i'^T\mathbf{k}')!}\prod_{(i_1,i_2)\in\mathbb{A}_2^{d}}\frac{1}{k^g_{i_1,i_2}!}\times\nonumber\\
%   \prod_{(i_1,i_2,i_3)\in\mathbb{A}_3^d}\frac{1}{k^g_{i_1,i_2,i_3}!}\times\hdots\times\prod_{(i_1,\hdots,i_{d'})\in\mathbb{A}_{d'}^d}\frac{1}{k^g_{i_1\hdots i_{d'}}!},  
% \end{align}
\begin{align}
  Q(\mathbf{k},\mathbf{k}')=\prod_{i=1}^{d}\frac{1}{(k_i-a_i'^T\mathbf{k}')!}\prod_{j=2}^{d'}\prod_{(i_1,\hdots,i_j)\in\mathbb{A}_j^{d}}\frac{1}{k^g_{i_1,\hdots, i_j}!},  
\end{align}
with $a_i'$ being the $i$-th row of the matrix $A'=[A_2\hdots A_{d'}]$ (see Appendix~\ref{sec:appx:B} for the derivation of the p.m.f.). Note that for $d'=1$, we have that $K$ is a collection of independent Poisson random variables, and when both $d=d'=1$, we recover the scalar Poisson distribution: 
\begin{align}
    \Pr(K=k) = \frac{e^{-\lambda}\lambda^k}{k!},\ \forall\ k\in\Nnot.
\end{align} 

%%%%%%%%%%%%%%%%%%%%%%%%%%%%%%%%%%%%%%%%%%
%%%%%%%%%%%%%%%%%%%%%%%%%%%%%%%%%%%%%%%%%%
%%%%%%%%% Definition of Multinomial %%%%%%
\textbf{Multinomial Distribution}: A $d$-dimensional random vector $\boldK$ is said to follow a multinomial distribution, i.e. $\boldK\sim\Multi(\mathbf{k}; n,\mathbf{p})$, if its p.m.f. is as follows:
\begin{align}
    \Pr(\boldK=\mathbf{k}) = \frac{n!}{\prod_{i=1}^d k_i!}\prod_{i=1}^d p_i^{k_i},\label{eq:poisson:multinomial_defn}
\end{align}
where $\mathbf{p}=\left[\begin{array}{ccc}
     p_1&\dots&p_d
\end{array}\right]^{T}$ is a probability vector such that $\sum_{i=1}^d p_i=1$, and $\mathbf{k}=\left[\begin{array}{ccc}
     k_1&\dots&k_d  
\end{array}\right]^{T}$ is a $d$-dimensional categorical vector such that $\sum_{i=1}^d k_i=n$. 
%Note that when $d=2$,~\eqref{eq:poisson:multinomial_defn} reduces to the binomial distribution~\cite{johnson1997discrete}.

\subsection{PID Background}\label{sec:PID_background}
Suppose $\msg$, $Y$, and $Z$ are random variables with joint distribution $P(\msg, Y, Z)$. According to~\cite{williams2010nonnegative, bertschinger2014quantifying}, there are three desirable equalities that should hold in a bivariate PID:
\begin{align}
    \label{eq:main_pid}
    I(\msg; Y, Z) &= UI(\msg; Y \backslash Z) + UI(\msg; Z \backslash Y) \nonumber
    \\
    &+ RI(\msg; Y, Z) + SI(\msg; Y, Z)
    \\
    I(\msg; Y) &= UI(\msg;Y \backslash Z) + RI(\msg; Y, Z) \label{eq:pid_aux1}
    \\
    I(\msg; Z) &= UI(\msg; Z \backslash Y) + RI(\msg; Y, Z). \label{eq:pid_aux2}
\end{align}
Here, $UI(\msg; Y \backslash Z)$ is the unique information $Y$ has about the message $\msg$ that is not in $Z$, $RI(\msg; Y, Z)$ is the redundant information shared between $Y$ and $Z$ about $\msg$, and $SI(M; Y, Z)$ is the synergistic information about $\msg$ that is attained through $Y$ and $Z$ jointly. Given~\eqref{eq:main_pid},~\eqref{eq:pid_aux1}, and~\eqref{eq:pid_aux2}, only one of $UI, RI$, or $SI$ need be defined to evaluate all four PID terms. Proposing a suitable measure is the focus of much PID research~\cite{williams2010nonnegative,harder2013bivariate,bertschinger2014quantifying,banerjee2018unique,kay2018exact,finn2018pointwise,niu2019measure,rosas2020operational,gurushankar2022extracting}. We refer the readers to~\cite{lizier2018information} for a review on PID.

Despite the diversity of proposed PID measures, several measures  \cite{harder2013bivariate, williams2010nonnegative, griffith2014quantifying, bertschinger2014quantifying} are in fact connected by Assumption $(\ast)$ of \cite{bertschinger2014quantifying}. This assumption states that UI should only depend on the marginals $P(\msg)$, $P(Y, \msg)$, and $P(Z, \msg)$, rather than the joint distribution $P(\msg, Y, Z)$. Recently, ``Blackwell sufficiency'' has been identified as another property \cite{venkatesh2022partial, venkatesh2023capturing}, distinct from Assumption $(\ast)$, connecting other PID measures \cite{bertschinger2014quantifying, banerjee2018unique, venkatesh2023capturing} based on the concept of sufficiency in statistical decision theory introduced by Blackwell \cite{blackwell1953equivalent}.

While Assumption $(\ast)$ makes intuitive sense, Blackwell sufficiency motivates an operational interpretation of the class of ``Blackwellian PIDs'' by giving the conditions for UI to go to zero, which was shown to be equivalent to stochastic degradedness of channels \cite{venkatesh2022partial}. The PID measure of \cite{bertschinger2014quantifying}, commonly referred to as BROJA-PID in the literature, is a Blackwellian PID that focuses on defining UI. While we illustrate our proofs with BROJA-PID (which is equivalently defined in \cite{griffith2014quantifying}), note that our results are applicable to Blackwellian PIDs in general.

The BROJA-PID defines UI as
\begin{align}
    UI(\msg; Y \backslash Z) = \min_{Q \in \Delta_P} I_Q(\msg; \YQ | \ZQ), \label{eq:ui_def}
\end{align}
where 
\begin{align}
    \Delta_P \myeq \{ Q(\msg, \YQ, \ZQ) \st Q(\msg, \YQ) &\myeq P(\msg, Y), \nonumber
    \\
    Q(\msg, \ZQ) &\myeq P(\msg, Z)\}.
\end{align}
Here, $I_Q(\msg; \YQ|\ZQ)$ is the conditional mutual information under the distribution $Q(\msg, \YQ, \ZQ)$. A useful property of Blackwellian PIDs is that the UI of one variable is zero if it is a stochastically degraded version of the other variable (Appendix B in \cite{venkatesh2022partial}). For BROJA-PID, this occurs when there is a Markov chain in $\Delta_P$ (Lemma 6 in \cite{bertschinger2014quantifying}). In general, there is not necessarily such a distribution in $\Delta_P$. However, we show that for the Poisson and multinomial systems there is indeed an appropriate Markov chain in $\Delta_P$. Finding the Markov chains in $\Delta_P$ for the Poisson and multinomial systems constitutes the essence of our proofs.

% We may equivalently define the problem in~\eqref{eq:ui_def} as
% \begin{align}
%     Q^{*}(\msg, Y, Z) =\argmin_{Q\in\Delta_P}I_Q(\msg; Y | Z).
% \end{align}

%but the value at the point of global minima (i.e. the value of UI) is unique for the optimization problem defined in~\eqref{eq:main_opt}, as it is a convex optimization problem~\cite{bertschinger2014quantifying}. Hence, finding any one $Q^*(\msg, Y, Z)$ suffices to deduce the value of the UI through~\eqref{eq:ui_def}.

%%%%%%%%%%%%%%%%%%%%%%%%%%%%%% SECTION %%%%%%%%%%%%%%%%%%%%%%%%%%%%%%
\section{Unique Information in the Poisson System}\label{sec:poisson_result}
In this section, we first introduce our definition of the Poisson system, which can be understood as a multivariate extension of the ``ideal Poisson channel'' used in~\cite{atar2012mutual}. Then, in Sec.~\ref{sec:poisson_theorem}, we derive sufficient conditions under which one of the UI terms in the PID of the Poisson system is zero.   

\subsection{Definition of the Poisson System}\label{sec:problem_setup}
Our definition of the Poisson system (briefly described in Sec.~\ref{sec:intro}) is inspired by the ``ideal Poisson channel'' discussed in~\cite{atar2012mutual} and~\cite{lapidoth1993reliability}. This is the canonical channel used to describe direct detection optical communication. The ideal Poisson channel is defined as $P(Y|\msg)=\Po(\gamma\msg)$, where the input and the output of the channel are $\msg$ and $Y$, respectively. Note that both $\msg$ and $Y$ are scalars, with $\msg$ being a non-negative random variable. Here, $\gamma$ plays the role of Signal-to-Noise Ratio (SNR). Intuitively, this can be understood as a linear scalar Poisson channel as $Y$ depends linearly on $\msg$ through its rate parameter. 

Our Poisson system provides an intuitive generalization of the ideal Poisson channel. This is because the ``main effects'' depend linearly on $\msg$, the ``two-way'' covariance terms (which can be thought of as product of two elements) depend linearly on $\msg^2$, and so on, culminating in the ``$d$-way covariance'' terms that depend linearly on $M^d$.
The reason for linear dependence on \textit{powers} of $M$ (i.e. $M$, $M^2$, ..., $M^d$) rather than just $M$ becomes apparent in our proof of Lemma~\ref{lemma:poisson:p(yg|y)}. This particular structure allows us to cancel out accumulated effects of $M$ resulting from consideration of higher order dependence. Thus, once these effects are removed, we can make the conclusion that at least one of the variables will have zero UI about the message.

Formally, we define the Poisson system as follows. Let $\msg$ be a non-negative random variable with p.d.f./p.m.f. $P(M)$, $\boldY$ be a $d_1$-dimensional random vector, and $\boldZ$ be a $d_2$-dimensional random vector. Define $P(\boldY|\msg)\myeq\Po(d_1,d_1', \boldsymbol{\Lambda}_{\boldY})$, and  $P(\boldZ|\msg)\myeq\Po(d_2,d_2', \boldsymbol{\Lambda}_{\boldZ})$, with: 
\begin{align}
    &\boldsymbol{\Lambda}_{\boldY} \myeq \begingroup\setlength\arraycolsep{1.6pt}\begin{bmatrix}
\lambda_1^{\boldY}&\dots&\lambda_{d_1}^{\boldY}&\lambda_{1,2}^{\boldY}&\dots&\lambda_{(d_1-1),d_1}^{\boldY}&\dots&\lambda_{d_1-(d_1'-1),\dots, d_1}^{\boldY}
\end{bmatrix}^T, \endgroup\nonumber\\
% &\lambda_{i_1}^{\boldY} = \gamma^{\boldY}_{i_1}\msg\ \forall\ i_1\in\mathbb{A}_1^{d_1'},\nonumber\\  &\lambda_{i_1i_2}^{\boldY}=\gamma^{\boldY}_{i_1i_2}\msg^2\ \forall\ (i_1,i_2)\in\mathbb{A}_2^{d_1'},\nonumber\\
% & \ \ \ \ \ \ \ \vdots\nonumber\\ 
&
\lambda_{i_1,\hdots, i_{j}}^{\boldY}=\gamma^{\boldY}_{i_1,\hdots, i_{j}}\msg^{j}\ \forall\ j\in[d_1]\text{ and }(i_1,\hdots,i_{j})\in\mathbb{A}_{j}^{d_1},\nonumber\\
&\boldsymbol{\Lambda}_{\boldZ} = \begingroup\setlength\arraycolsep{1.6pt}
\begin{bmatrix}
\lambda_1^{\boldZ}&\dots&\lambda_d^{\boldZ}&\lambda_{1,2}^{\boldZ}&\dots&\lambda_{(d_2-1),d_2}^{\boldZ}&\dots&\lambda_{d_2-(d_2'-1),\dots, d_2 
}^{\boldZ}
\end{bmatrix}^T, \endgroup\nonumber\\
% &\lambda_{i_1}^{\boldZ} = \gamma^{\boldZ}_{i_1}\msg\ \forall\ i\in\mathbb{A}_1^{d_2'},\nonumber\\  &\lambda_{i_1i_2}^{\boldZ}=\gamma^{\boldZ}_{i_1i_2}\msg^2\ \forall\ (i_1,i_2)\in\mathbb{A}_2^{d_2'},\nonumber\\
% & \ \ \ \ \ \ \ \vdots\nonumber\\ 
&
\lambda_{i_1,\hdots, i_{j}}^{\boldZ}=\gamma^{\boldZ}_{i_1,\hdots, i_{j}}\msg^{j}\ \forall\ j\in[d_2]\text{ and }(i_1,\hdots, i_{j})\in \mathbb{A}_{j}^{d_2}.\label{eq:poisson_system}
 \end{align}
Here, the parameters
\begin{align}
    \gamma_1^{\boldY},\hdots,\gamma_{d_1-(d_1'-1),\hdots, d_1}^{\boldY},\gamma_1^{\boldZ},\hdots,\gamma_{d_2-(d_2'-1),\hdots, d_2}^{\boldZ}\in\Rplus, \nonumber
\end{align}
{and can be thought as SNR terms analogously to the ideal Poisson channel}. Furthermore, let $A_{\boldY}$ and $A_{\boldZ}$ be the corresponding $A$-matrices, defined in~\eqref{eq:poisson:dependent_defn}, associated with $P(\boldY|M)$ and $P(\boldZ|M)$, respectively. {Intuitively, $P(\boldY|\msg)$ and $P(\boldZ|\msg)$ can be thought as a collection of $d_1$ and $d_2$ dependent ideal Poisson channels, respectively, with the structure of dependency as specified in~\eqref{eq:poisson_system}.}

\subsection{UI in the Poisson System}\label{sec:poisson_theorem}
In this section, we derive the sufficient conditions under which one of the UI terms is zero for the Poisson system introduced in Sec.~\ref{sec:problem_setup}. Our proof consists of showing the existence of a Markov chain of the form $\msg\rightarrow\boldYQ\rightarrow\boldZQ$ such that its marginals are the same as $P(\msg,\boldY)$ and $P(\msg,\boldZ)$ and it satisfies the set of conditions~\eqref{eq:theorem:poisson:cond:1} in Theorem~\ref{theorem:1:poisson_depend}.

{For the special case of $P(\boldY|\msg)$ and $P(\boldZ|\msg)$ consisting of conditionally independent ideal Poisson channels (see Corollary~\ref{cor:ui_indep_poisson}), our result can be interpreted as the Poisson equivalent for the well known result for additive Gaussian channels, where the channel with the overall lower SNR is just a stochastically degraded version of the channel with an overall higher SNR~\cite{shang2012noisy} (for the special case of \textit{scalar} $\msg$). This result follows a trend where many results that are known to be true for additive Gaussian channels are also true for the ideal Poisson channel~\cite{verdu1999poisson,atar2012mutual}. Theorem~\ref{theorem:1:poisson_depend}, derives equivalent conditions on the SNR, i.e.~\eqref{eq:theorem:poisson:cond:1}, for the more general case where the ideal Poisson channels in $P(\boldY|\msg)$ and $P(\boldZ|\msg)$ are not independent. In other words, the conditions in~\eqref{eq:theorem:poisson:cond:1} ensure that the inequality between the SNR terms of the Poisson channels holds for \textit{all} orders of dependencies, each associated to $M^1, \dots, M^d$ as discussed in Sec.~\ref{sec:problem_setup}.

Now, we provide a brief proof overview of Theorem~\ref{theorem:1:poisson_depend}. Theorem~\ref{theorem:1:poisson_depend} show the existence of the desired Markov chain $\msg\rightarrow\boldYQ\rightarrow\boldZQ$ by constructing a larger Markov chain:
$\msg\rightarrow \boldYQ\rightarrow \boldY^g\rightarrow \boldZ^g\rightarrow\boldZQ$, having the same marginals, $P(M,\boldY)$ and $P(M,\boldZ)$. Here, $\boldY^g$ and $\boldZ^g$ are intermediate variables that are analogous to the $\boldK^g$ defined in Sec.~\ref{sec:defn}, which we marginalize to obtain the desired Markov chain $\msg \rightarrow\boldYQ\rightarrow\boldZQ$. Note that to construct the Markov chain $\msg\rightarrow \boldYQ\rightarrow \boldY^g\rightarrow \boldZ^g\rightarrow\boldZQ$, it suffices to specify the five marginal distributions $\Tilde{Q}(\msg)$, $\Tilde{Q}(\boldYQ|\msg)$, $\Tilde{Q}(\boldY^g|\boldYQ)$, $\Tilde{Q}(\boldZ^g|\boldY^g)$, and $\Tilde{Q}(\boldZQ|\boldZ^g)$ due to the Markov structure. 

Theorem~\ref{theorem:1:poisson_depend} utilizes Lemmas~\ref{lemma:poisson_indep:1},~\ref{lemma:poisson:p(yg|y)}, and ~\ref{lemma:poisson_gen_indep}, so we first provide a brief discussion of these lemmas. First, in Lemma~\ref{lemma:poisson_indep:1}, we show that the existence of the aforementioned Markov chain indeed implies one of the UI terms is zero for the PID defined in~\cite{bertschinger2014quantifying}.
\begin{lemma}\label{lemma:poisson_indep:1}
    If there is a Markov chain $Q^*$ of the form $\msg\rightarrow \boldYQ\rightarrow \boldZQ$ in $\Delta_P$ for the optimization problem~\eqref{eq:ui_def}, then $\msg\rightarrow\boldYQ\rightarrow\boldZQ$ minimizes~\eqref{eq:ui_def}, and $UI(\msg; \boldZ \backslash \boldY)=I_{Q^*}(\msg;\boldZQ|\boldYQ)=0$.
\end{lemma}
\begin{proof}
    See Appendix~\ref{appx:proof of lemma 1}.
\end{proof} 

Lemma~\ref{lemma:poisson:p(yg|y)} is used in Theorem~\ref{theorem:1:poisson_depend} for defining $\Tilde{Q}(\boldY^g|\boldYQ)$ to construct the Markov chain $\msg\rightarrow\boldYQ\rightarrow\boldY^g$. The proof of Lemma~\ref{lemma:poisson:p(yg|y)} comprises of showing that the distribution $\Tilde{Q}(\boldY^g|\boldYQ,\msg)\myeq \Tilde{Q}(\boldY^g|\boldYQ)$. We use Bayes' Theorem to express $\Tilde{Q}(\boldY^g|\boldYQ,\msg)=\Tilde{Q}(\boldYQ|\boldY^g,\msg)\Tilde{Q}(\boldY^{g}|\msg)/\Tilde{Q}(\boldYQ|\msg)$, and show the right hand side of the equation does not contain any terms depending on $M$, hence obtaining the required result. 
% comes from showing that $Q(\boldY^{g}|\boldY)$ exploiting the dependence structure of $\boldsymbol{\Lambda}_{\boldY}$ specified in~\eqref{eq:poisson_system}, and the fact that $\boldY=A\boldY^g$. Combining Here, $\boldY^g$ is analogous to $\boldK^g$ discussed in Sec.~\ref{sec:defn}. 
\begin{lemma}\label{lemma:poisson:p(yg|y)}
    Let $\boldY$ be a $d$-dimensional vector and $\msg$ be a positive random variable, i.e. $\msg\sim P(\msg), P(\msg\leq 0)=0$. Let $P(\boldY|\msg)= \Po(d,d', \boldsymbol{\Lambda})$, where:
    \begin{align}
&\boldsymbol{\Lambda}=\begingroup 
\setlength\arraycolsep{2pt}
\begin{bmatrix}
    \lambda_1&\hdots&\lambda_d&\lambda_{1,2}&\hdots\lambda_{d-(d'-1),\hdots, d}
\end{bmatrix}^T\endgroup,\nonumber\\
% &\lambda_{i_1} = \gamma_{i_1}\msg\ \forall\ i_1\in\mathbb{A}_1^{d},\nonumber\\  &\lambda_{i_1i_2}=\gamma_{i_1i_2}\msg^2\ \forall\ (i_1,i_2)\in\mathbb{A}_2^{d},\nonumber\\
% & \ \ \ \ \ \ \ \vdots\nonumber\\ 
% &
&\lambda_{i_1,\hdots, i_{j}}=\gamma_{i_1,\hdots, i_{j}}\msg^{j}\ \forall\ j\in[d']\text{ and } (i_1,\hdots,i_{d'})\in\mathbb{A}_{j}^{d}.
    \end{align}
    Let $\boldY^g=\begingroup\setlength\arraycolsep{1.6pt}\begin{bmatrix}
    Y_1^g&\dots&Y_d^g&Y_{1,2}^g&\dots&Y_{(d-1),d}^g&\dots&Y_{d-(d'-1),\dots, d}^g
\end{bmatrix}^T\endgroup$, where $P(\boldY^g|M)=\Po(d_{\boldsymbol{\Lambda}},1,\boldsymbol{\Lambda})$
% \begin{align*}
%     &P(Y_{i_1}^g|\msg) = \Po(\gamma^{\boldY}_{i_1}\msg)\ \forall\ i_1\in\mathbb{A}_1^{d},\nonumber\\  &P(Y_{i_1i_2}^g|\msg)=\Po(\gamma^{\boldY}_{i_1i_2}\msg^2) \forall\ (i_1,i_2)\in\mathbb{A}_2^{d},\nonumber\\
% & \ \ \ \ \ \ \ \vdots\nonumber\\ 
% &
% P(Y_{i_1\hdots i_{d'}}^g|\msg)=\Po(\gamma^{\boldY}_{i_1\hdots i_{d'}}\msg^{d'})\ \forall\ (i_1,\hdots,i_{d'})\in\mathbb{A}_{d'}^{d},
% \end{align*}
% all the elements of $\boldY^g$ are independent conditioned on $\msg$ 
and $\boldY=A\boldY^g$, where $d_{\boldsymbol{\Lambda}}$ is the dimension of $\boldsymbol{\Lambda}$, and $A=[A_1\hdots A_{d'}]$ as defined in Sec.~\ref{sec:defn}. Then $P(\boldY^g|\boldY,\msg)=P(\boldY^g|\boldY)$, i.e. $\msg$, $\boldY^g$ and $\boldY$ form the following Markov chain $\msg\rightarrow\boldY\rightarrow\boldY^g$.
\end{lemma}
\begin{proof}
    See Appendix~\ref{appx:proof of lemma 2}.
\end{proof}
Lemma~\ref{lemma:poisson_gen_indep} is  used in Theorem~\ref{theorem:1:poisson_depend} to extend the Markov chain $\msg\rightarrow\boldYQ\rightarrow\boldY^g$ to $\msg\rightarrow\boldYQ\rightarrow\boldY^g\rightarrow\boldZ^g$. For this extension of the Markov chain, we provide an explicit construction of $\Tilde{Q}(\boldY^g|\boldZ^g)$, where $\Tilde{Q}(\boldY^g|\boldZ^g)$ consists of a product of Multinomial distributions. The proof of Lemma~\ref{lemma:poisson_gen_indep} comes from a multivariate extension of a well-known result regarding Poisson variables, namely if $Y\sim \Po(\lambda)$, and $Z|Y\sim\text{Binomial}(Y,p)$, then $Z\sim \Po(p\lambda)$~\cite{chatterji1963some}. 

\begin{lemma}\label{lemma:poisson_gen_indep}
 Let $\msg$ be as defined in Sec.~\ref{sec:problem_setup}. Define, $\boldY=\begingroup\setlength\arraycolsep{1.6pt}\begin{bmatrix}
\boldY_{1}^T&\dots&\boldY_{d_1}^T
\end{bmatrix}^T\endgroup$ and $\boldZ=\begingroup\setlength\arraycolsep{1.6pt}\begin{bmatrix}
    \boldZ_1&\dots&\boldZ_{d_2}^T
\end{bmatrix}^T\endgroup$, where $\boldY_i=\begingroup\setlength\arraycolsep{1.6pt}\begin{bmatrix}
    Y_{i,1}&\dots&Y_{i,n_i}
\end{bmatrix}^T\endgroup,\boldZ_j=\begingroup\setlength\arraycolsep{1.6pt}\begin{bmatrix}
    Z_{j,1}&\dots&Y_{j,m_j}
\end{bmatrix}^T\endgroup$ are vectors of size $n_i,m_j\in\N$, respectively, $\forall\ i\in[d_1],\ j\in[d_2]$. Let, $P(Y_{i,j}|\msg) = \Po(\gamma^{\boldY}_{i,j}\msg^{i})\ \forall\ (i,j)\in\{(i,j) \st i\in[d_1],j\in[n_i]\}$, and $P(Z_{i,j}|\msg) = \Po(\gamma^{\boldZ}_{i,j}\msg^{i})\ \forall\ (i,j)\in\{(i,j) \st i\in[d_2],j\in[m_i]\}$. Furthermore, let all elements of $\boldY$ be mutually conditionally independent of each other (conditioned on $\msg$). Similarly, let all the elements of $\boldZ$ be mutually conditionally independent of each other (conditioned on $\msg$). If $d_1\geq d_2$ and~\eqref{eq:lemma:poisson_gen_indep:1} hold:
\begin{align}
 &\sum_{j=1}^{n_i}\gamma^{\boldY}_{i,j}\geq  \sum_{j=1}^{m_i}\gamma^{\boldZ}_{i,j}\ \forall\ i\in[d_2]\label{eq:lemma:poisson_gen_indep:1}
\end{align}
Then there is a distribution in $\Delta_P$ of the form 
    \begin{align}
        \Tilde{Q}(\msg,\boldYQ, \boldZQ) = P(\msg)P(\boldY|\msg)\Tilde{Q}(\boldZQ|\boldYQ).
        \label{eq:poisson:lemma2:1234}
    \end{align}
\end{lemma}
\begin{proof}
    See Appendix~\ref{appx:proof of lemma 3}.
\end{proof}
Now, we discuss Theorem \ref{theorem:1:poisson_depend}, which provides the sufficient conditions under which one of the UI terms in the Poisson system is zero.

\begin{theorem}\label{theorem:1:poisson_depend}
    If $\msg,\boldY$ and $\boldZ$ are defined as in Sec.~\ref{sec:problem_setup}, with $d_1'\geq d_2'$, and the following conditions hold: 
    \begin{align}
        % \sum_{i_1\in\mathbb{A}_1^{d_1'}}\gamma_{i_1}^{\boldY}&\geq         \sum_{i_1\in\mathbb{A}_1^{d_2'}}\gamma_{i_1}^{\boldZ},\nonumber\\
        % \sum_{(i_1,i_2)\in\mathbb{A}_2^{d_1'}}\gamma_{i_1i_2}^{\boldY}&\geq         \sum_{(i_1,i_2)\in\mathbb{A}_2^{d_2'}}\gamma_{i_1i_2}^{\boldZ},\nonumber\\
        % &\vdots\nonumber\\
        \sum_{(i_1,\hdots, i_{j})\in\mathbb{A}_{j}^{d_1}}\gamma_{i_1,\hdots, i_{j}}^{\boldY}&\geq         \sum_{(i_1,\hdots,i_{j})\in\mathbb{A}_{j}^{d_2}}\gamma_{i_1,\hdots, i_{j}}^{\boldZ}\ \forall\ j\in[d_2'],\label{eq:theorem:poisson:cond:1}
    \end{align}
    then there exists a Markov chain $\msg\rightarrow\boldYQ\rightarrow\boldZQ$ that lies in $\Delta_P$. Consequently, $UI(\msg; \boldZ \backslash \boldY)=0$.
\end{theorem}
\begin{proof}
% We show the existence of the Markov chain $\msg\rightarrow\boldY\rightarrow\boldZ$ having the same marginals $P(M,\boldY)$ and $P(M,\boldZ)$ in the set $\Delta_{P}$ by constructing a larger Markov chain:
% $\msg\rightarrow \boldY\rightarrow \boldY^g\rightarrow \boldZ^g\rightarrow\boldZ$, having the same marginals: $P(M,\boldY)$ and $P(M,\boldZ)$. Here, $\boldY^g$ and $\boldZ^g$ are intermediate variables, which we can marginalize to obtain the desired Markov chain $X\rightarrow\boldY\rightarrow\boldZ$. 
%Denote the joint distribution of $\msg\rightarrow \boldY\rightarrow \boldY^g\rightarrow \boldZ^g\rightarrow\boldZ$ as $\Tilde{Q}(M,\boldY,\boldY^g,\boldZ^g,\boldZ)$.
% In order to construct the Markov chain $\msg\rightarrow \boldY\rightarrow \boldY^g\rightarrow \boldZ^g\rightarrow\boldZ$, it suffices to specify the marginal distributions $\Tilde{Q}(\msg)$, $\Tilde{Q}(\boldY|\msg)$, $\Tilde{Q}(\boldY^g|\boldY)$, $\Tilde{Q}(\boldZ^g|\boldY^g)$, and $\Tilde{Q}(\boldZ|\boldZ^g)$ due to its Markovian structure. 
Let us provide the explicit construction of the Markov chain $\msg\rightarrow \boldYQ\rightarrow \boldY^g\rightarrow \boldZ^g\rightarrow\boldZQ$ having the marginals $P(\msg,\boldY)$, and $P(\msg,\boldZ)$. For $\Tilde{Q}(\msg)$ and $\Tilde{Q}(\boldYQ|\msg)$, we choose them to be equal to $P(\msg)$ and $P(\boldY|\msg)$, i.e. $\Tilde{Q}(\msg)=P(\msg)$, and $\Tilde{Q}(\boldYQ|\msg)=P(\boldY|\msg)$. Note that due to this construction, $P(\msg,\boldY)=\Tilde{Q}(\msg,\boldYQ)$ holds trivially. 

For $\Tilde{Q}(\boldY^g|\boldYQ)$, we use the result described in Lemma~\ref{lemma:poisson:p(yg|y)}. Note that the construction for $\Tilde{Q}(\boldY^g|\boldYQ)$ is not explicit but rather implicit. Let $d_{\boldsymbol{\Lambda}_{\boldY}}$ be the dimension of $\boldsymbol{\Lambda}_{\boldY}$, then we explicitly choose
\begin{align}
    \Tilde{Q}(\boldYQ|\msg) = \Po(d_1,d_1',\mathbf{\Lambda}_{\boldY}), \nonumber
    \\
    \Tilde{Q}(\boldYQ|\boldY^g,\msg)=\delta_{K}(\boldYQ=A^{\boldY}\boldY^g), \nonumber
    \\
    \Tilde{Q}(\boldY^g|\msg) = \Po(d_{\boldsymbol{\Lambda}_{\boldY}},1,\mathbf{\Lambda}_{\boldY}), \nonumber
\end{align}
 and derive $\Tilde{Q}(\boldY^g|\boldYQ,\msg)$ through Bayes' Theorem. Here, $\delta_{K}(\cdot)$ is the Kronecker delta function~\cite{oppenheim1997signals}. By Lemma~\ref{lemma:poisson:p(yg|y)}, we know that $\Tilde{Q}(\boldY^g|\boldYQ,\msg)=\Tilde{Q}(\boldY^g|\boldYQ)$, and hence we have the Markov chain $\msg\rightarrow\boldYQ\rightarrow\boldY^g$. 
%and $\boldY^g=\begingroup\setlength\arraycolsep{1.6pt}\begin{bmatrix}
% Y_1^g&\dots&Y_{d_1'}^g&Y_{12}^g&\dots&Y_{(d_1'-1)d_1'}^g&\dots&Y_{d_1'-(d_1-1)\dots d_1'}^g
% \end{bmatrix}^T\endgroup$

For choosing $\Tilde{Q}(\boldZ^g|\boldY^g)$, we rely on the result of Lemma~\ref{lemma:poisson_gen_indep}. Let us define the random vectors:
\begin{align*}
\boldY_j &\myeq
\left[Y_{i_1,\hdots, i_j}^g\right]^T_{(i_1,\hdots,i_j)\in\mathbb{A}_j^{d_1}},
\\
\boldZ_j &\myeq
\left[Z_{i_1,\hdots, i_j}^g\right]^T_{(i_1,\hdots,i_j)\in\mathbb{A}_j^{d_2}},    
\end{align*}
 i.e. $\boldY_j$ and $\boldZ_j$ are random vectors containing all terms of the form $Y_{i_1,\hdots,i_j}^g$, and $Z_{k_1,\hdots,k_j}^g$, where $(i_1,\hdots,i_j)\in\mathbb{A}_{j}^{d_1}$ and $(k_1,\hdots,k_j)\in\mathbb{A}_{j}^{d_2}$, respectively. Note that we can write $\boldY^g\myeq\begingroup\setlength\arraycolsep{1.6pt}\begin{bmatrix}
\boldY_1^T&\dots&\boldY_{d_1'}^T\end{bmatrix}^T\endgroup$, and we define $\boldZ^g=\begingroup\setlength\arraycolsep{1.6pt}\begin{bmatrix}
\boldZ_1^T&\dots&\boldZ_{d_2'}^T\end{bmatrix}^T\endgroup$. Then, we construct $\Tilde{Q}(\boldZ^g|\boldY^g)$ as a product of $d_2$  multinomial distributions, described below:  
\begin{align*}
    \Tilde{Q}(\boldZ^g|\boldY^g) = \prod_{j=1}^{d_2}\Tilde{Q}(\boldZQ_j|\boldYQ_j),
\end{align*}
where $\Tilde{Q}(\boldZQ_j\myeq \mathbf{z}_j|\boldYQ_j\myeq\mathbf{y}_j)=\Multi(\mathbf{k}_j;N_j,\mathbf{p}_i)$, and
\begin{align*}
&N_j=\mathbf{1}^T\mathbf{y}_j,\ \mathbf{k}_j=\begingroup \setlength\arraycolsep{1.6pt}\begin{bmatrix}
    \mathbf{z}_j^T&\mathbf{1}^T\mathbf{y}_j- \mathbf{1}^T\mathbf{z}_j
    \end{bmatrix}^T\endgroup,\\
&\mathbf{p}_i=\begingroup \setlength\arraycolsep{1.6pt}\begin{bmatrix}
      p_{1,\hdots, j} &\cdots& p_{d_2-(j-1),\hdots, d_2} &1- \sum_{(i_1,\hdots,i_j)\in\mathbb{A}_{j}^{d_2}}p_{i_1,\hdots, i_j} \end{bmatrix}^T\endgroup,\\
        &p_{i_1,\hdots, i_j}=\frac{\gamma_{i_1,\hdots, i_j}^{\boldZ}}{\sum_{(i_1,\hdots,i_j)\in\mathbb{A}_{j}^{d_1}}\gamma_{i_1,\hdots, i_j}^{\boldY}}\ \forall\ (i_1,\hdots,i_j)\in\mathbb{A}_j^{d_2}.
\end{align*}
By construction, $\boldY^g$ consists of mutually independent Poisson random variables conditional on $\msg$, and the condition of Lemma~\ref{lemma:poisson_gen_indep} (specified in~\eqref{eq:lemma:poisson_gen_indep:1}) is satisfied in the assumption of Theorem \ref{theorem:1:poisson_depend}, i.e. equation~\eqref{eq:theorem:poisson:cond:1}. Therefore, after marginalizing $\boldYQ$ out of the Markov chain $\msg \rightarrow \boldY \rightarrow \boldY^g \rightarrow \boldZ^g$, we use the result of Lemma~\ref{lemma:poisson_gen_indep} on $\msg \rightarrow \boldY^g \rightarrow \boldZ^g$ to conclude that $\Tilde{Q}(\boldZ^g|\msg)=\Po(1,d_{\boldsymbol{\Lambda}_{\boldZ}},\boldsymbol{\Lambda}_{\boldZ})$, where $d_{\boldsymbol{\Lambda}_{\boldZ}}$ is the dimension of $\boldsymbol{\Lambda}_{\boldZ}$.

%  Similarly, $\boldZ^g=\begingroup\setlength\arraycolsep{1.6pt}\begin{bmatrix}
%     Z_1^g&\dots&Z_{d_2'}^g&Z_{12}^g&\dots&Z_{(d_2'-1)d_2'}^g&\dots&Z_{d_2'-d_2\dots d_2'}^g
% \end{bmatrix}^T\endgroup$ is also a collection of conditionally independent Poisson random variables, i.e. $P(\boldZ^g|\msg)\myeq\prod_{i=1}^{l_{\boldZ}}\Po(\mathbf{e}_i^T\boldsymbol{\Lambda}_{\boldZ})$, where $l_{\boldZ}$ is the dimension of $\boldsymbol{\Lambda_{\boldZ}}$. $P(\boldZ|\msg)\myeq\Po(d_2,\boldsymbol{\Lambda_{\boldZ}})$, as described in Sec.~\ref{sec:defn}. 

We choose $\Tilde{Q}(\boldZQ|\boldZ^g)$ as the following deterministic transformation $\boldZQ=A_{\boldZ}\boldZ^g$ to obtain the Markov chain $\msg\rightarrow\boldYQ\rightarrow\boldY^g\rightarrow\boldZ^g\rightarrow\boldZQ$. Marginalizing $\boldYQ$ and $\boldY^g$ in the above Markov chain, we get the following Markov chain: $\msg\rightarrow\boldZ^g\rightarrow\boldZQ$. Now, since $\boldZQ=A_{Z}\boldZ^g$ and $P(\boldZ^g|\msg)=\Po(d_{\boldsymbol{\Lambda}_{\boldZ}},1,\boldsymbol{\Lambda}_{\boldZ})$, $P(\boldZ|\msg)= \Po(d_2,d_2', \Lambda_{\boldZ})$ (by definition of the multivariate Poisson described in Sec.~\ref{sec:defn}). Since we have $\Tilde{Q}(\msg)=P(\msg)$ and $\Tilde{Q}(\boldZQ|\msg)=P(\boldZ|\msg)$, we also have $\Tilde{Q}(\msg,\boldZQ)=P(\msg,\boldZ)$. 

Now, marginalizing $\boldY^g$ and $\boldZ^g$ from the Markov chain $\msg\rightarrow\boldYQ\rightarrow\boldY^g\rightarrow\boldZ^g\rightarrow\boldZQ$, we obtain the desired Markov chain $\msg\rightarrow\boldYQ\rightarrow\boldZQ$ having $\Tilde{Q}(\msg,\boldYQ)\myeq P(\msg,\boldY)$ and  $\Tilde{Q}(\msg,\boldZQ)\myeq P(\msg,\boldZ)$. By Lemma~\ref{lemma:poisson_indep:1}, we know that the proposed $\Tilde{Q}(\msg,\boldYQ,\boldZQ)$ minimizes~\eqref{eq:ui_def}, and also has $I_{\Tilde{Q}}(\msg;\boldZQ|\boldYQ)=0$, concluding our proof.
% Now, since we have a Markov chain $\msg\rightarrow\boldYQ\rightarrow\boldY^g\rightarrow\boldZ^g\rightarrow\boldZQ$, such that its marginals $Q(\msg,\boldYQ)=P(\msg,\boldY)$ and  $Q(\msg,\boldZQ)=P(\msg,\boldZ)$, we can marginalize $\boldY^g$ and $\boldZ^g$ out of the above Markov chain to get the desired Markov chain $\msg\rightarrow\boldYQ\rightarrow\boldZQ$ that lies in $\Delta_P$. Using the result of Lemma~\ref{lemma:poisson_indep:1}, we know that the proposed $Q(\msg,\boldYQ,\boldZQ)$ minimizes~\eqref{eq:ui_def}, and also has $I_{\Tilde{Q}}(\msg;\boldZQ|\boldYQ)=0$, concluding our proof.
\end{proof}
\begin{corollary}\label{cor:ui_indep_poisson}
    Suppose $M, \boldY$, and $\boldZ$ are defined as in Section \ref{sec:problem_setup}, with $d_1'=d_2'=1$. If
    \begin{align}
        \sum_{i=1}^{d_1}\gamma_{i}^{\boldY}  \myleq \sum_{i=1}^{d_2} \gamma_{i}^{\boldZ}, \nonumber
    \end{align}
    then $UI(M; \boldZ \backslash \boldY) = 0$.
\end{corollary}
\begin{proof}
    This follows from Theorem \ref{theorem:1:poisson_depend} when $d_1' \myeq d_2' \myeq 1$.
\end{proof}
{Corollary~\ref{cor:ui_indep_poisson} states that if $\boldY$ and $\boldZ$ consist of conditionally independent ideal Poisson channels, then the channel with overall lower SNR (measured as the sum of individual SNR's) is the stochastic degraded version of the one with the higher overall SNR.}
%%%%%%%%%%%%%%%%%%%%%%%%%%%%%%%%%%%%%%%%%%%%%%%%%%%%%
%%%%%%%%%%%%%%%%%%%%%%%%%%%%%%%%%%%%%%%%%%%%%%%%%%%%%
%%%%%%%%%%%%%%%%%%%%%%%%%%%%%%%%%%%%%%%%%%%%%%%%%%%%%

%%%%%%%%%%%%%%%%%%%%%%%%%%%%%% SECTION %%%%%%%%%%%%%%%%%%%%%%%%%%%%%%
\section{Unique Information in the Multinomial System}
\label{sec:multinomial}
We show a parallel result for UI in the multinomial system, i.e. at least one of $Y$ or $Z$ has zero UI about $\msg$. In the multinomial system, $\msg \sim P(\msg)$ with support $\Nnot$, and $Y$ and $Z$ are multinomially thinned versions of $\msg$, i.e. $P(Y|\msg) \myeq \text{Multinomial}(\mathbf{y}; m, \mathbf{p_y})$ and $P(Z|\msg) \myeq \text{Multinomial}(\mathbf{z}; m, \mathbf{p_z})$, where $\mathbf{p_y}$ and $\mathbf{p_z}$ are probability vectors of size $s_y$ and $s_z$, respectively. For the multinomial system, we invoke the same construction-based proof technique as in the Poisson system described in Sec. \ref{sec:poisson_theorem}. Lemma \ref{lemma:multinom_conditional} is used to provide the explicit construction of  $\Tilde{Q}(\ZQ|\YQ)$ used to create the Markov chain $\msg\rightarrow \YQ \rightarrow \ZQ$.
\begin{lemma}\label{lemma:multinom_conditional}
    Suppose $Y \sim \text{Multinomial}(\mathbf{y}; n, \mathbf{p})$ and $Z|Y \sim \text{Multinomial}\left(\mathbf{z}; \sum_{j \in \idxset} Y_{j}, \mathbf{q}\right)$ have $\ky \myeq |\mathbf{p}|$ and $\kz \myeq |\mathbf{q}|$ number of classes, respectively, where $\idxset$ is an arbitrary set of class indices of $Y$ and $|\idxset| \leq \ky$. Then $Z \sim \text{Multinomial}\left(\mathbf{z}; n, \mathbf{q}^* \right)$, where 
    \begin{align} \label{eq:multinomial_distribution}
        {q_{i}}^* = 
        \begin{cases}
            \left[\sum_{j \in \idxset} p_{j} \right] \cdot q_i, &\ 1 \leq i < k_2
            \\
            1- \left[\sum_{j \in \idxset} p_{j} \right] \cdot \left[ \sum_{j=1}^{k_2-1} q_j  \right], &\ i = k_2.
        \end{cases}
    \end{align}
\end{lemma}
\begin{proof}
    See Appendix~\ref{appx:proof-lemma-multinomial}.
\end{proof}
To obtain the UI result for the multinomial system, we show that we can construct a valid Markov distribution that lies in $\Delta_P$. Hence, one of the UI terms is zero.
\begin{theorem}\label{theorem:multinom_pid}
    Suppose $\msg,Y,$ and $Z$ are as defined in the multinomial system described above. Let $I_Y \myeq [s_y]$ and $I_Z \myeq [s_z]$ be the set of class indices of $Y$ and $Z$, respectively. If 
    \begin{align}\label{eq:multin_assumption}
        \min_{i \in I_Z} p_{z_i} \geq \min_{i \in I_Y} p_{y_i},
    \end{align}
    then there is a Markov chain in $\Delta_P$. Thus, $UI(M; Z \backslash Y) = 0$.
\end{theorem}
\begin{proof}
We define $\Tilde{Q}(\msg, \YQ, \ZQ) \myeq P(\msg) P(Y|\msg) \Tilde{Q}(\ZQ|\YQ)$ with $\Tilde{Q}(\ZQ|\YQ) \myeq \text{Multinomial}\left(\mathbf{z}; \sum_{i \in I_Y^-} y_i, \mathbf{p_z}^* \right)$. Here, $I_Y^- \myeq I_Y \backslash \argmin_{i \in I_Y} p_{y_i}$ is the set of class indices of $Y$ excluding the index of the smallest probability element, and 
\begin{align}
    \mathbf{p_z}^* \myeq \setlength\arraycolsep{1.6pt}\begin{bmatrix}
        \frac{p_{z_1}}{\sum_{i \in I_Y^-} p_{y_i}}  &\hdots&\frac{p_{z_{s_z-1}}}{\sum_{i \in I_Y^-} p_{y_i}}&1- \frac{\sum_{i = 1}^{\sizemultb-1} p_{z_i}} {\sum_{i \in I_Y^-} p_{y_i}}
    \end{bmatrix}^T.
\end{align}
We first show the proposed $\Tilde{Q}(\ZQ|\YQ)$ is a valid distribution. By assumption~\eqref{eq:multin_assumption},
\begin{align}
    &\min_{i \in I_Z} p_{z_i} \geq \min_{i \in I_Y} p_{y_i}
    \\
    &\Rightarrow \sum_{i \in I_Y^-} p_{y_i} \geq \sum_{i \in I_Z^-} p_{z_i} \geq p_{z_i}, \; \forall i \in I_Z
    \\
    &\Rightarrow 1 \geq \frac{p_{z_i}}{\sum_{i \in I_Y^-} p_{y_i}} \geq 0,  \ \forall i \in I_Z,
\end{align}
where the last inequality is due to $\mathbf{p_y}$ and $\mathbf{p_z}$ forming valid multinomial distributions. Following the arguments made in the proof of Lemma \ref{lemma:indep_poisson:2}, we then show the two relevant pairwise marginals of $\Tilde{Q}$ are equal to those of the true $P$. The equality $\Tilde{Q}(\msg, \YQ) \myeq P(\msg, Y)$ is given by equation~\eqref{eq:poisson:lemma2:finaleq4} in Lemma \ref{lemma:indep_poisson:2}, while $\Tilde{Q}(\msg, \ZQ) \myeq P(\msg, Z)$ is given by Lemma \ref{lemma:multinom_conditional} and~\eqref{eq:poisson:lemma2:tildeq(x|y)_defn} of Lemma \ref{lemma:indep_poisson:2}.
% Next, we show that $\Tilde{Q}(\msg, Y, Z) \in \Delta_P$ by noting that both pairwise marginal constraints are satisfied. 
% For $\msg = \msgsmall$, let $\mathcal{Y} = \{ \mathbf{y} \in \Nnot^{\sizemulta-1} \st \sum_{i \in I_Y} y_i = \msgsmall \}$ and $\mathcal{Z} = \{ \mathbf{z} \in \Nnot^{\sizemultb-1} \st \sum_{i \in I_Z} z_i = \msgsmall \}$ be the support of $Y$ and $Z$, respectively. \question{Should this part of the proof be repeated from Lemma 2?}
% \begin{align}
%     \Tilde{Q}(\msg = \msgsmall, Y = \boldy) 
%     &= \sum_{\boldz \in \mathcal{Z}} P(\msgsmall) P(\boldy|\msgsmall) \Tilde{Q}(\boldz|\boldy)
%     \\
%     &= P(\msgsmall, \boldy)
% \end{align}
% \begin{align}
%     \Tilde{Q}(\msg = \msgsmall, Z = \boldz)
%     &= P(\msgsmall) \sum_{\boldy \in \mathcal{Y}} P(\boldy | \msgsmall) \Tilde{Q}(\boldz | \boldy)
%     \\
%     &= P(\msgsmall) \sum_{\boldy \in \mathcal{Y}} P(\boldy | \msgsmall) \Tilde{Q}(\boldz | \boldy, \msgsmall)
%     \label{eq:multinom_cond_indep},
% \end{align}
% where~\eqref{eq:multinom_cond_indep} is due to the Markov chain construction. Finally, by Lemma~\ref{lemma:multinom_conditional}, the sum over $\boldy \in \mathcal{Y}$ is also multinomial-distributed:
% \begin{align}
%     \Tilde{Q}(Z = \boldz | \msg = \msgsmall) &= \text{Multinomial} \left( \boldz; n, \mathbf{p_z}^* \right)
%     \\
%     &= P(Z = \boldz | \msg = \msgsmall).
% \end{align}
Hence, $\Tilde{Q}(\msg, \YQ, \ZQ) \in \Delta_P$ and is a minimizer of~\eqref{eq:ui_def} by Lemma \ref{lemma:poisson_indep:1}, and $UI(M; Z \backslash Y) \myeq 0$.
\end{proof}
% We find a useful result in the special case when $Y$ and $Z$ are binomial conditional on $\msg$.
% \begin{corollary}\label{cor:ui_binomial}
%     Suppose $\msg \sim P(\msg)$ with support $\Nnot$ and $P(Y|\msg) = \text{Binomial}(n, p_y)$ and $P(Z|\msg) = \text{Binomial}(n, p_z)$. If $p_y > p_z$ then $UI(\msg; Z \backslash Y) = 0$.
% \end{corollary}
% \begin{proof}
%     Follows from Theorem \ref{theorem:multinom_pid} when $\sizemulta = \sizemultb = 2$.
% \end{proof}

%%%%%%%%%%%%%%%%%%%%%%%%%%%%%% SECTION %%%%%%%%%%%%%%%%%%%%%%%%%%%%%%
\section{Discussion and Limitations}\label{sec:discussion}
% We study the bivariate PID for two systems: when both agents are (1) Poisson-distributed, and (2) multinomial-distributed given the message. We define a multivariate Poisson channel model to characterize the full dependency structure and show that UI exists in at most one agent when the conditions of~\eqref{eq:theorem:poisson:cond:1} are satisfied. In the multinomial system, UI is positive in a single agent when the inequality of~\eqref{eq:theorem:multinomial:cond:1} is strict, otherwise, neither agent has UI. Corollaries \ref{cor:ui_indep_poisson} and \ref{cor:ui_binomial} highlight the ease of estimation in the special cases of element-wise Poisson independence and binomials, respectively, where a trivial comparison of Poisson or binomial parameters suffices to determine the agent(s) with zero UI.
We study the bivariate PID for the Poisson system and the multinomial system. We provide sufficient conditions for both systems under which one of the UI terms in their respective PID is zero, thereby facilitating the computation and estimation of the remaining PID terms. {Our results can alternatively be interpreted as sufficient conditions for checking the stochastic degradeness between two Poisson channels (defined in Sec.~\ref{sec:problem_setup}) and two multinomial channels (defined in Sec.~\ref{sec:multinomial}).} 

In this work, our results are restricted to \emph{scalar} $\msg$ for both systems. Future work would address the extension to a vector $\msg$, as was done for the Gaussian case in \cite{venkatesh2022partial}. There also exist alternative definitions of multivariate Poisson distributions~\cite{johnson1997discrete}, which can be considered for constructing the Poisson system. The multivariate Poisson distribution used in this work only allows for positive covariance, limiting its modelling capabilities.
% Corollaries \ref{cor:ui_indep_poisson} and \ref{cor:ui_binomial} highlight the ease of estimation in the special cases of element-wise Poisson independence and binomials, respectively, where a trivial comparison of Poisson or binomial parameters suffices to determine the agent(s) with zero UI.

Notably, given that there exist three distinct systems (Poisson, multinomial, and Gaussian) for scalar $\msg$ in which one of the UI terms is zero, an interesting future direction is discovering the commonality between these systems to obtain a more general characterization of systems in which UI can be reduced to zero. 

\newpage
\enlargethispage{-1.4cm} 
%%%%%%

%%%%%
%%% Do not uncomment for double-blind review submission!
% \section*{Acknowledgments}
%
% We are indebted to Michael Shell for maintaining and improving
% \texttt{IEEEtran.cls}. 

%%%%%%%%%%%%%%%%%%%%%%%%%%%%%% SECTION %%%%%%%%%%%%%%%%%%%%%%%%%%%%%%
% \section{Conclusion}
% \label{sec:conclusion}

% We conclude by pointing out that on the last page the columns need to
% be balanced. Instructions for that purpose are given in the source
% file (they are commented out).

%%%%%%
%% To balance the columns at the last page of the paper use this
%% command somewhere at the top of the first column of the last page:
%%
% \enlargethispage{-5cm} 
%%
%% where the exact amount of page reduction has to be adapted to the
%% actual situation.
%%
%% If the balancing should occur in the middle of the references, use
%% the following trigger:
%%
% \IEEEtriggeratref{3}
%%
%% which triggers a \newpage (i.e., new column) just before the given
%% reference number. Note that you need to adapt this if you modify
%% the paper. The "triggered" command can be changed if desired:
%%
% \IEEEtriggercmd{\enlargethispage{-20cm}}
%%
%%%%%%

%%%%%%
%% References:
%% We recommend the usage of BibTeX:
%%
\bibliographystyle{IEEEtran}
\bibliography{Bibliography}
%%
%% where we here have assume the existence of the files
%% definitions.bib and bibliofile.bib.
%% BibTeX documentation can be obtained at:
%% http://www.ctan.org/tex-archive/biblio/bibtex/contrib/doc/
%%%%%%
%% Or you use manual references (pay attention to consistency and the
%% formatting style!):

% \begin{thebibliography}{9}

% \bibitem{Laport:LaTeX}
% L.~Lamport,
%   \emph{\LaTeX: A Document Preparation System,} 
%   Addison-Wesley, Reading, Massachusetts, USA, 2nd~ed., 1994. 

% \bibitem{GMS:LaTeXComp}
% F.~Mittelbach, M,~Goossens, J.~Braams, D.~Carlisle, and
% C.~Rowley, \emph{The {\LaTeX} Companion,} Addison-Wesley,
% Reading, Massachusetts, USA, 2nd~ed., 2004.

% \bibitem{oetiker_latex}
% T.~Oetiker, H.~Partl, I.~Hyna, and E.~Schlegl, \emph{The Not So Short
%   Introduction to {\LaTeX2e}}, version 6.4, Mar.~9, 2021. [Online].
%   Available: \url{https://tobi.oetiker.ch/lshort/}

% \bibitem{typesetmoser}
% S.~M. Moser, \emph{How to Typeset Equations in {\LaTeX}}, version 4.6,
%   Sep. 29, 2017. [Online]. Available:
%   \url{https://moser-isi.ethz.ch/manuals.html#eqlatex}

% \bibitem{shell15}
% M.~Shell, ``How to use the {IEEEtran} {\LaTeX} class,'' \emph{Journal of
%   {\LaTeX} Class Files}, vol.~14, no.~8, Aug. 2015. [Online]. Available:
%   \url{https://mirrors.ctan.org/macros/latex/contrib/IEEEtran/IEEEtran\_HOWTO.pdf}

% \bibitem{IEEE:AuthorToolbox}
% IEEE, \emph{Author Center.} [Online.] Available:
%   \url{https://ieeeauthorcenter.ieee.org/}

% \end{thebibliography}

%%%%%% 
%% Appendix:
%% If needed a single appendix is created by
%%
\newpage
% \appendix
%%
%% If several appendices are needed, then the command
%%
\appendices
%%
%% in combination with further \section-commands can be used.
%%%%%%

%%%%%%%%%%%%%%%%%%%%%%%%%%%%%%%%%%%%%%%%%%%%%%%%%%%%%%%%%%%%%%%%%%
%%%%%%%%%%%%%%%%%%%%%%%%%%%%%%%%%%%%%%%%%%%%%%%%%%%%%%%%%%%%%%%%%%
%%%%%%%%%%%%%%%%%%%%%%%%%%%%%%%%%%%%%%%%%%%%%%%%%%%%%%%%%%%%%%%%%%

%%%%%%%%%%%%%%%%%%%%%%%%%%%%%%%%%%%%%%%%%%%%%%%%%%%%%
%%%%%%%%%%%%%%%%%%%%%%% Lemma %%%%%%%%%%%%%%%%%%%%%%%
%%%%%%%%%%%%%%%%%%%%%%%%%%%%%%%%%%%%%%%%%%%%%%%%%%%%%
\section{}
\subsection{Proof of Lemma~\ref{lemma:poisson_indep:1} }\label{appx:proof of lemma 1}
%%%%%%%%%%%%%%%%%%%%%%%%%%%%%%%%%%%%%%%%%%%%%%%%%%%%%
%%%%%%%%%%%%%%%%%%%%%%% Lemma %%%%%%%%%%%%%%%%%%%%%%%
%%%%%%%%%%%%%%%%%%%%%%%%%%%%%%%%%%%%%%%%%%%%%%%%%%%%%
% \begin{lemma}\label{lemma:poisson_indep:1}
%     For the optimization problem defined in~\eqref{eq:main_opt}, if $\exists$ a Markov chain of the form $\msg\rightarrow \boldY\rightarrow \boldZ$ in $\Delta_P$, then $\msg\rightarrow\boldY\rightarrow\boldZ$ minimizes~\eqref{eq:main_opt} or alternatively is a $Q^*(\msg,\boldY,\boldZ)$. 
% \end{lemma}
\begin{proof}
By the definition of the Markov chain we know that $\msg$ and $\boldZQ$ are conditionally independent given $\boldYQ$. Hence, for the Markov chain $\msg\rightarrow\boldYQ\rightarrow\boldZQ$, we have:
\begin{align}
    I_{\msg\rightarrow\boldYQ\rightarrow\boldZQ}(\msg;\boldZQ|\boldYQ) = 0,
\end{align}
where $I_{\msg\rightarrow\boldYQ\rightarrow\boldZQ}(\msg;\boldZQ|\boldYQ)$ is the conditional mutual information between $\msg$ and $\boldZQ$ given $\boldYQ$ for the Markov chain $\msg\rightarrow\boldYQ\rightarrow\boldZQ$. Since $I_Q(\msg;\boldZQ|\boldYQ)\geq 0\ \forall\ Q\in\Delta_P$, the Markov chain $\msg\rightarrow\boldYQ\rightarrow\boldZQ$ achieves the minimum for~\eqref{eq:ui_def}.
%the optimization problem $\argmin_{Q\in\Delta_P}I_Q(X;\boldZ|\boldY)$.
\end{proof}
\subsection{Proof of Lemma~\ref{lemma:poisson:p(yg|y)}}\label{appx:proof of lemma 2}
\begin{proof}
    To prove the above lemma, all we need to show is that $P(\boldY^g|\boldY,\msg)$ does not depend on $\msg$. We first calculate $P(\boldY^g|\boldY,\msg)$. By Bayes' Theorem, we know:
    \begin{align}
        P(\boldY^g|\boldY,\msg) = \frac{P(\boldY^g|\msg)P(\boldY|\boldY^g,\msg)}{P(\boldY|\msg)}\label{eq:pyg|X:0}
    \end{align}
    Let us now write the expression for $P(\boldY|\msg)$ using~\eqref{eq:poisson_multivariate:pmf}:
\begin{align}
    P(\boldY|\msg)\myeq& Q(\mathbf{y},\mathbf{y}')e^{-\mathbf{1}^T\boldsymbol{\Lambda}}\prod_{i=1}^{d}(\gamma_i \msg)^{y_i}\times\\&\sum_{\mathbf{y'}\in C}\prod_{j=2}^{d'}\prod_{(i_1,\hdots,i_j)\in\mathbb{A}_{j}^d}\left(\frac{\gamma_{i_1,\hdots, i_j}\msg^j}{\prod_{l=1}^{j}\gamma_{i_l}\msg}\right)^{y_{i_1,\hdots, i_j}'}.
\end{align}    
Canceling $\msg$ in the above equation for the terms inside the summation, we get:
\begin{align}
P(\boldY|\msg)\myeq& Q(\mathbf{y},\mathbf{y}')e^{-\mathbf{1}^T\boldsymbol{\Lambda}}\prod_{i=1}^{d}(\gamma_i \msg)^{y_i}\times\\&\sum_{\mathbf{y'}\in C}\prod_{j=2}^{d'}\prod_{(i_1,\hdots,i_j)\in\mathbb{A}_{j}^d}\left(\frac{\gamma_{i_1,\hdots, i_j}}{\prod_{l=1}^{j}\gamma_{i_l}}\right)^{y_{i_1,\hdots, i_j}'}.\label{eq:poisson_cancel}
\end{align}
Absorbing all terms that do not depend upon $\msg$ into $B$, i.e.: 
\begin{align}
  B = \sum_{\mathbf{y'}\in C}\prod_{j=2}^{d'}\prod_{(i_1,\hdots,i_j)\in\mathbb{A}_{j}^d}\left(\frac{\gamma_{i_1,\hdots, i_j}}{\prod_{l=1}^{j}\gamma_{i_l}}\right)^{y_{i_1,\hdots, i_j}'},\label{eq:}
\end{align}
 Then, we can rewrite~\eqref{eq:poisson_cancel} as: 
\begin{align}
    P(\boldY|\msg) &= Be^{-\mathbf{1}^T\boldsymbol{\Lambda}}\prod_{i=1}^{d}(\gamma_i \msg)^{y_i},\\
    &= Be^{-\mathbf{1}^T\boldsymbol{\Lambda}}\msg^{\sum_{i=1}^{d}y_i}\prod_{i=1}^{d}(\gamma_i)^{y_i},\\
    &\overset{(a)}{=}Be^{-\mathbf{1}^T\boldsymbol{\Lambda}}\msg^{\sum_{i=1}^{d}y_i}\overset{(b)}{=}Be^{-\mathbf{1}^T\boldsymbol{\Lambda}}\msg^{\mathbf{1}^T\mathbf{y}},\label{eq:p(yg|X):1}
\end{align}
where, in (a) we further absorb $\prod_{i=1}^{d}(\gamma_i)^{y_i}$ into $B$ since it does not depend upon $\msg$, and in (b) we substitute $\sum_{i=1}^d y_i=\mathbf{1}^T\mathbf{y}$. Similarly, let us write out the expression for $P(\boldY^g|\msg)$:
\begin{align}
    P(\boldY^g|&\msg)=e^{-\mathbf{1}^T\boldlambda}\prod_{j=1}^{d'}\prod_{(i_1,\hdots,i_j)\in\mathbb{A}_j^{d}}(\gamma_{i_1,\hdots, i_j}\msg^j)^{y^g_{i_1,\hdots, i_j}}.
\end{align}
Collecting all the $\msg$ terms, and
% \begin{align}
%     &P(\boldY^g|\msg)=e^{-\mathbf{1}^T\boldlambda}\prod_{j=1}^{d'}\prod_{(i_1,\hdots,i_j)\in\mathbb{A}_j^{d}}(\gamma_{i_1\hdots i_j})^{y^g_{i_1\hdots i_j}}\times\nonumber\\
% &\msg^{\sum_{j=1}^{d'}\sum_{(i_1,\hdots,i_j)\in\mathbb{A}_j^{d}}\left(jy_{i_1\hdots i_j}^g\right)}.
% \end{align}
absorbing all the terms that do not depend upon $\msg$ into $D$, we obtain:
\begin{align}
    &P(\boldY^g|\msg)=De^{-\mathbf{1}^T\boldlambda}\msg^{\sum_{j=1}^{d'}\sum_{(i_1,\hdots,i_j)\in\mathbb{A}_j^{d}}\left(jy_{i_1,\hdots, i_j}^g\right)}.\label{eq:p(yg|X):2}
\end{align}
Now let us analyze the term $\mathbf{1}^TA\mathbf{y}^g$: 
\begin{align*}
  \mathbf{1}^TA\mathbf{y}^g\overset{(a)}{=}(\mathbf{y}^g)^TA^T\mathbf{1} \overset{(b)}{=}   (\mathbf{y}^g)^T\begin{bmatrix}
      A_1^T\mathbf{1}\\
      \vdots\\
      A_{d'}^T\mathbf{1}\end{bmatrix}\overset{(c)}{=}(\mathbf{y}^g)^T\begin{bmatrix}
      \mathbf{1}\\
      2\mathbf{1}\\
      \vdots\\
      d'\mathbf{1}
  \end{bmatrix},
\end{align*}
where (a) uses the fact that $\mathbf{1}^TA\mathbf{y}^g$ is a scalar and hence is equal to its transpose, (b) uses the fact $A\myeq [A_1\hdots A_{d'}]$, and (c) follows from the special structure of $A_i$, i.e. that each column only contains $i$ ones and $d-i$ zeros, and the fact that $A_1^T\mathbf{1}$ is akin to summing up the columns, hence $A_i^T\mathbf{1}=i\mathbf{1}$. Equivalently we can rewrite the above equation as:
\begin{align}
   \mathbf{1}^TA\mathbf{y}^g
   %=& \sum_{i_1\in\mathbb{A}_1^{d}}y_{i_1}^g+\hdots+\sum_{(i_1,\hdots,i_{d'})\in\mathbb{A}_{d'}^{d}}d'y_{i_1,\hdots,i_{d'}}^g,\nonumber\\
   =&\sum_{j=1}^{d'}\sum_{(i_1,\hdots,i_j)\in\mathbb{A}_j^{d}}\left(jy_{i_1,\hdots, i_j}^g\right).\label{eq:p(yg|X):3}
\end{align}
Substituting~\eqref{eq:p(yg|X):3} into~\eqref{eq:p(yg|X):2}:
\begin{align}
    &P(\boldY^g|\msg)=De^{-\mathbf{1}^T\Lambda}\msg^{\mathbf{1}^TA\mathbf{y}^g}\label{eq:p(yg|X):4}
\end{align}
Now, let us write out the expression for $P(\boldY|\boldY^g, \msg)$. Since $\boldY=A\boldY^g$, $P(\boldY|\boldY^g)$ can be represented as a Kronecker delta function with the condition $\boldY=A\boldY^g$, i.e.
\begin{align}
    P(\boldY=\mathbf{y}|\boldY^g=\mathbf{y}^g, \msg)&=    P(\boldY=\mathbf{y}|\boldY^g=\mathbf{y}^g)\nonumber\\ &= \delta_{K}\left(\mathbf{y}=A\mathbf{y}^g\right),\label{eq:p(yg|X):5}
\end{align}
where $\delta_K(\cdot)$ is the Kronecker delta function.
Substituting~\eqref{eq:p(yg|X):1},~\eqref{eq:p(yg|X):4} and~\eqref{eq:p(yg|X):5} in~\eqref{eq:pyg|X:0}, we get: 
\begin{align}
 P(\boldY^g|\boldY,\msg) &= \frac{De^{-\mathbf{1}^T\Lambda}\msg^{\mathbf{1}^TA\mathbf{y}^g}\delta_K(\mathbf{y}=A\mathbf{y}^g)}{Be^{-\mathbf{1}^T\boldsymbol{\Lambda}}\msg^{\mathbf{1}^T\mathbf{y}}},\\
&\overset{(a)}{=} \frac{D\msg^{\mathbf{1}^TA\mathbf{y}^g}\delta_K(\mathbf{y}=A\mathbf{y}^g)}{B\msg^{\mathbf{1}^T\mathbf{y}}},\\
&\overset{(b)}{=} \frac{D\msg^{\mathbf{1}^TA\mathbf{y}^g}\delta_K(\mathbf{y}=A\mathbf{y}^g)}{B\msg^{\mathbf{1}^TA\mathbf{y}^g}},\\
 &\overset{(c)}{=}\frac{D\delta_K(\mathbf{y}=A\mathbf{y}^g)}{B},
\end{align}
where we obtain (a) by canceling the term $e^{-\mathbf{1}^T\boldsymbol{\Lambda}}$, (b) by
using the fact $\mathbf{y}=A\mathbf{y}^g$ due to the delta function, and (c) by canceling the term $M^{\mathbf{1}^TA\mathbf{y}^g}$. Since the terms $D,B$ and $\delta_K(\mathbf{y}=A\mathbf{y}^g)$ do not depend upon $\msg$, we can conclude $P(\boldY^g|\boldY,\msg)$ also does not depend upon $\msg$, i.e. $P(\boldY^g|\boldY,\msg)=P(\boldY^g|\boldY)$, or $\msg\rightarrow\boldY\rightarrow\boldY^g$.
\end{proof}
\subsection{Proof of Lemma~\ref{lemma:poisson_gen_indep} }\label{appx:proof of lemma 3}
\begin{proof}
This Lemma immediately follows as a consequence of Lemma~\ref{lemma:indep_poisson:2} described below. To show that that $\Tilde{Q}(\msg,\boldYQ,\boldZQ)$ lies in $\Delta_P$, all we need to show is $\Tilde{Q}(\msg,\boldYQ)=P(\msg,\boldY)$, and $P(\msg,\boldZQ)=\Tilde{Q}(\msg,\boldZ)$. The first equality, i.e. $\Tilde{Q}(\msg,\boldYQ)=P(\msg,\boldY)$ follows trivially from construction (for a more detailed argument, see the proof of Lemma~\ref{lemma:indep_poisson:2}). To show the second equality, i.e. $P(\msg,\boldZ)=\Tilde{Q}(\msg,\boldZQ)$, it suffices to show $P(\boldZ|\msg)=\Tilde{Q}(\boldZQ|\msg)$, as $P(\msg)=\Tilde{Q}(\msg)$ by construction. Let the support of $\boldY_i$ be $\mathcal{Y}_i$, and the support of $\boldY$ be $\mathcal{Y}$. By Law of Total Probability and using the fact that $\Tilde{Q}(\boldYQ|\msg)=P(\boldY|\msg)$, we know
\begin{align}
\Tilde{Q}(\boldZQ|\msg)&=\sum_{\mathbf{y}\in\mathcal{Y}}P(\boldY=\mathbf{y}|\msg)\Tilde{Q}(\boldZQ|\boldYQ=\mathbf{y}).\label{eq:lemma:poisson_gen_indep:2}
\end{align}
Now, since $\boldY=\begingroup\setlength\arraycolsep{1.6pt}\begin{bmatrix}
\boldY_{1}^T&\dots&\boldY_{d_1}^T
\end{bmatrix}^T\endgroup$, where $\{\boldY_i\}_{i=1}^{d_1}$ are mutually conditionally independent, we have $P(\boldY=\mathbf{y}|\msg)=\prod_{i=1}^{d_1}P(\boldY_i=\mathbf{y}_i|\msg)$. Substituting this fact in~\eqref{eq:lemma:poisson_gen_indep:2}, we get:
\begin{align}
\Tilde{Q}(\boldZQ|\msg)&=\sum_{\mathbf{y}\in\mathcal{Y}}\left[\prod_{i=1}^{d_1}P(\boldY_i=\mathbf{y}_i|\msg)\right]\Tilde{Q}(\boldZQ|\boldYQ=\mathbf{y}).\label{eq:lemma:poisson_gen_indep:3}
\end{align}
Using the fact that $\Tilde{Q}(\boldZQ|\boldYQ=\mathbf{y})=\prod_{i=1}^{d_2}\Tilde{Q}(\boldZQ_i|\boldYQ_i=\mathbf{y_i})$ in~\eqref{eq:lemma:poisson_gen_indep:3}, we get:
\begin{align}
\Tilde{Q}(\boldZQ|\msg)&=\sum_{\mathbf{y}\in\mathcal{Y}}\prod_{i=1}^{d_1}P(\boldY_i=\mathbf{y}_i|\msg)\prod_{i=1}^{d_2}\Tilde{Q}(\boldZQ_i|\boldYQ_i=\mathbf{y}_i).\label{eq:lemma:poisson_gen_indep:4}
\end{align}
Now, combining the two products in the equation~\eqref{eq:lemma:poisson_gen_indep:4}, we obtain:
\begin{align}
\Tilde{Q}(\boldZQ|\msg)&=\sum_{\mathbf{y}\in\mathcal{Y}}\prod_{i=1}^{d_1}P(\boldY_i=\mathbf{y}_i|\msg)\Tilde{Q}(\boldZQ_i|\boldYQ_i=\mathbf{y}_i),\label{eq:lemma:poisson_gen_indep:5}
\end{align}
where $\Tilde{Q}(\boldZQ_i|\boldYQ_i=\mathbf{y}_i)=1\ \forall\ i>d_2$. Now, the above expression can be equivalently expressed as:
\begin{align}
\Tilde{Q}(\boldZQ|\msg)&=\prod_{i=1}^{d_1}\sum_{\mathbf{y}_i\in\mathcal{Y}_i}P(\boldY_i=\mathbf{y}_i|\msg)\Tilde{Q}(\boldZQ_i|\boldYQ_i=\mathbf{y}_i),\label{eq:lemma:poisson_gen_indep:6}  
\end{align}
Now, for $i>d_2$, we know that $\Tilde{Q}(\boldZQ_i|\boldYQ_i=\mathbf{y}_i)=1$, so the following term reduces as
\begin{align}
    \sum_{\mathbf{y}_i\in\mathcal{Y}_i}P(\boldY_i=\mathbf{y}_i|\msg)\Tilde{Q}(\boldZQ_i|\boldYQ_i=\mathbf{y}_i) \nonumber 
    \\
    =\sum_{\mathbf{y}_i\in\mathcal{Y}_i}P(\boldY_i=\mathbf{y}_i|\msg)=1, \nonumber
\end{align}
since we are summing over a probability distribution. Hence, we can reduce~\eqref{eq:lemma:poisson_gen_indep:6} as follows:
\begin{align}
\Tilde{Q}(\boldZQ|\msg)&=\prod_{i=1}^{d_2}\underbrace{\sum_{\mathbf{y}_i\in\mathcal{Y}_i}P(\boldY_i=\mathbf{y}_i|\msg)\Tilde{Q}(\boldZQ_i|\boldYQ_i=\mathbf{y}_i)}_{=\Tilde{Q}(\boldZQ_i|\msg)}. \label{eq:lemma:poisson_gen_indep:7}  
\end{align}
Now, note that $P(\boldY_i|\msg)$ is a collection of mutually conditionally independent Poisson random vectors having rates of the form $\gamma_{i,j}^{\boldY}\msg^n$, and $\Tilde{Q}(\boldZQ_i|\boldYQ_i)$ is a Multinomial distribution with parameters: 
\begin{align*}
&N_i=\sum_{j=1}^{n_i}y_{i,j},
\\ 
&\mathbf{k}_i=\begingroup \setlength\arraycolsep{1.6pt}\begin{bmatrix}
z_{i,1}&\cdots&z_{i,m_i}&\sum_{j=1}^{n_i}y_{i,j}-\sum_{j=1}^{m_i}z_{i,j}
\end{bmatrix}^T\endgroup,\\
&\mathbf{p}_i=\begingroup \setlength\arraycolsep{1.6pt}\begin{bmatrix}
    \frac{\gamma_{i,1}^{\boldZ}}{\sum_{j=1}^{n_i}\gamma_{i,j}^{\boldY}}&\cdots&        \frac{\gamma_{i,m_i}^{\boldZ}}{\sum_{j=1}^{n_i}\gamma_{i,j}^{\boldY}}&1-\frac{\sum_{j=1}^{m_i}\gamma_{i,j}^{\boldZ}}{\sum_{j=1}^{n_i}\gamma_{i,j}^{\boldY}}    \end{bmatrix}^T\endgroup.\\
\end{align*}
Hence, using the result of Lemma~\ref{lemma:indep_poisson:2}, we know that $\tilde{Q}(\boldZQ_i|\msg)=\prod_{j=1}^{m_i}\Po(\gamma_{i,j}^{\boldZ}\msg^i)$. Furthermore, note that $P(\boldZ|\msg)=\prod_{i=1}^{d_2}P(\boldZ_i|\msg)$, where $P(\boldZ_i|\msg)=\prod_{j=1}^{m_i}\Po(\gamma_{i,j}^{\boldZ}\msg^i)$. Hence, we have $P(\boldZ_i|\msg)=\Tilde{Q}(\boldZQ_i|\msg)=\sum_{\mathbf{y}_i\in\mathcal{Y}_i}P(\boldY_i=\mathbf{y}_i|\msg)\Tilde{Q}(\boldZQ_i|\boldYQ_i=\mathbf{y}_i)$. Substituting this result in~\eqref{eq:lemma:poisson_gen_indep:7}, we obtain:
\begin{align}
\Tilde{Q}(\boldZQ|\msg)&=\prod_{i=1}^{d_2}P(\boldZ_i|\msg)=P(\boldZ|\msg),\label{eq:lemma:poisson_gen_indep:8}
\end{align}
which completes the proof.
\end{proof}
\begin{lemma}\label{lemma:indep_poisson:2}
  Let $\msg\in\mathbb{R}$, $\msg\sim P(\msg)$ and $P(\msg\leq 0)=0$. Define, $\boldY\in\mathbb{R}^{d_1}$ and $\boldZ\in\mathbb{R}^{d_2}$. Furthermore, define 
  \begin{align}
      P(\boldY|\msg)\myeq\prod_{i\in[d_1]}P(Y_i|\msg)\myeq\prod_{i\in[d_1]}\Po(\gamma_{i}^{\boldY}\msg^{n}), \nonumber
      \\
      P(\boldZ|\msg)\myeq\prod_{i\in[d_2]}P(Z_i|\msg)\myeq\prod_{i\in[d_2]}\Po(\gamma_{i}^{\boldZ}\msg^{n}) \nonumber
  \end{align}
  for some $n\in\N$. If $\sum_{i=1}^{d_1}\gamma_{i}^{\boldY}\geq\sum_{i=1}^{d_2} \gamma_{i}^{\boldZ}$, then the distribution $\Tilde{Q}(\msg,\boldYQ,\boldZQ)$ defined in~\eqref{eq:poisson:lemma2:1} lies in $\Delta_P$.
    \begin{align}
        \Tilde{Q}(\msg,\boldYQ, \boldZQ) = P(\msg)P(\boldY|\msg)\Tilde{Q}(\boldZQ|\boldYQ),\label{eq:poisson:lemma2:1}
    \end{align}
    where $\Tilde{Q}(\boldZQ|\boldYQ)$ is a $\Multi(\mathbf{k};N,\mathbf{p})$ distribution with:
    \begin{align}
     N&=\sum_{i=1}^{d_1}y_i, \nonumber
     \\ 
     \mathbf{k}&=\begin{bmatrix}
        z_1&\cdots&z_{d_2}&\sum_{i=1}^{d_1}y_i-\sum_{i=1}^{d_2}z_i
    \end{bmatrix}^T, \nonumber\\   
    \mathbf{p}&=\begin{bmatrix}
        \frac{\gamma_{1}^{\boldZ}}{\sum_{i=1}^{d_1}\gamma_{i}^{\boldY}}&\cdots&        \frac{\gamma_{d_2}^{\boldZ}}{\sum_{i=1}^{d_1}\gamma_{i}^{\boldY}}&1-\frac{\sum_{i=1}^{d_2}\gamma_{i}^{\boldZ}}{\sum_{i=1}^{d_1}\gamma_{i}^{\boldY}}    \end{bmatrix}^T.\label{eq:poisson:lemma2:qz|y_defn}
    \end{align}
% The full expression for $\Tilde{Q}(\boldZ=\mathbf{z}|\boldY=\mathbf{y})$ is as follows:
%     \begin{align}
%         \Tilde{Q}(\boldZ=\mathbf{z}|\boldY=\mathbf{y}) =& \frac{\left(\sum_{i=1}^{d_1}y_i\right)!}{\left(\sum_{i=1}^{d_1}y_i\myminus\sum_{i=1}^{d_2}z_i\right)!\prod_{i=1}^{d_2}z_i!}\nonumber\\
%         \times \prod_{i=1}^{d_2}\left(\frac{\gamma_{i}^{\boldZ}}{\sum_{j=1}^{d_1}\gamma_{i}^{\boldY}}\right)^{z_i}&\left(1-\frac{\sum_{i=1}^{d_2}\gamma_{i}^{\boldZ}}{\sum_{i=1}^{d_1}\gamma_{i}^{\boldY}}\right)^{\sum_{i=1}^{d_1}y_i-\sum_{i=1}^{d_2}z_i}.\label{eq:poisson:lemma2:qz|y_defn}
%     \end{align}
\end{lemma}
\begin{proof}[Proof of Lemma \ref{lemma:indep_poisson:2}]
    To show $\Tilde{Q}(\msg,\boldYQ,\boldZQ)$ lies in $\Delta_P$, we need to show the following two equalities:
    \begin{align}
        \Tilde{Q}(\msg,\boldYQ) = P(\msg,\boldY) = P(\msg)P(\boldY|\msg),\label{eq:poisson:lemma2:2}\\
        \Tilde{Q}(\msg,\boldZQ) = P(\msg,\boldZ) = P(\msg)P(\boldZ|\msg).\label{eq:poisson:lemma2:3}
    \end{align}
     Showing the equality in~\eqref{eq:poisson:lemma2:2} is trivial as  by construction $\Tilde{Q}(\msg)=P(\msg)$, and $\Tilde{Q}(\boldYQ|M)=P(\boldY|\msg)$ which implies:
    \begin{align}
%         \Tilde{Q}(\msg,\boldY) &= \sum_{\mathbf{z}}P(\msg)P(\boldY|\msg)\Tilde{Q}(\boldZ=\mathbf{z}|\boldY), \label{eq:poisson:lemma2:my_equality1} \\
% &=P(\msg)P(\boldY|\msg)\sum_{\mathbf{z}}\Tilde{Q}(\boldZ=\mathbf{z}|\boldY) \label{eq:poisson:lemma2:my_equality2} \\
%         \intertext{$\sum_{\mathbf{z}}\Tilde{Q}(\boldZ=\mathbf{z}|\boldY)=1$, as we are summing a probability distribution.}
        \Tilde{Q}(\msg,\boldYQ) =\Tilde{Q}(\msg)\Tilde{Q}(\boldYQ|\msg)= P(\msg)P(\boldY|\msg).\label{eq:poisson:lemma2:finaleq4}
    \end{align}
    To show the equality in~\eqref{eq:poisson:lemma2:3}, let us calculate $\Tilde{Q}(\msg,\boldZQ)$:
    \begin{align}
        \Tilde{Q}(\msg,\boldZQ) &= \sum_{\mathbf{y}}P(\msg)P(\boldY=\mathbf{y}|\msg)\Tilde{Q}(\boldZQ|\boldYQ=\mathbf{y}),\\
        \intertext{Taking $P(\msg)$ out from the summation in the above equation:}        
        \Tilde{Q}(\msg,\boldZQ)&=P(\msg)\sum_{\mathbf{y}}P(\boldY=\mathbf{y}|\msg)\Tilde{Q}(\boldZQ|\boldYQ=\mathbf{y})\label{eq:poisson:lemma2:tildeq(x|y)_defn}
    \end{align}
Let us focus on the term $\sum_{\mathbf{y}}P(\boldY\myeq\mathbf{y}|\msg)\Tilde{Q}(\boldZQ|\boldYQ\myeq\mathbf{y})$. Since $\Tilde{Q}(\boldZQ|\boldYQ)\myeq\Multi(\mathbf{k};N,\mathbf{p})$, any $\mathbf{z}$ such that $\sum_{i=1}^{d_2}z_i>\sum_{i=1}^{d_1}y_i$ would have a probability of $0$. This implies that in the summation $\sum_{\mathbf{y}}P(\boldY|\msg)\Tilde{Q}(\boldZQ|\boldYQ)$, instead of summing over all possible $\mathbf{y}$, we should only sum over all $\mathbf{y}$ which satisfy the following inequality $\sum_{i=1}^{d_1}y_i\geq\sum_{i=1}^{d_2}z_i$ ($\Tilde{Q}(\boldZQ|\boldYQ)$ would be zero for all $\mathbf{y}$, where $\sum_{i=1}^{d_1}y_i<\sum_{i=1}^{d_2}z_i$). Denote the set $\mathcal{Y}_{\mathbf{z}}=\left\{\mathbf{y}\in\Nnot^{d_1}\st\sum_{i=1}^{d_1}y_i\geq \sum_{i=1}^{d_2}z_i\right\}$. Then, we have:
\begin{align}
&\sum_{\mathbf{y}}P(\boldY|\msg)\Tilde{Q}(\boldZQ|\boldYQ) = \sum_{\mathbf{y} \in \mathcal{Y}_{\mathbf{z}}}P(\boldY|\msg)\Tilde{Q}(\boldZQ|\boldYQ)\label{eq:poisson:lemma2:equivalence}
\end{align}
Now, since $\boldY$ is just a collection of mutually conditionally independent Poisson random variables, we can write $P(\boldY|\msg)$ as a product of Poisson distributions:
\begin{align}
P(\boldY|\msg)&\myeq\prod_{i=1}^{d_1}\Po\left(y_i;\gamma_{i}^{\boldY}\msgsmall^{n}\right)\nonumber\\ &= e^{-\msgsmall^n\sum_{i=1}^{d_1}\gamma_{i}^{\boldY}}\msgsmall^{n\sum_{i=1}^{d_1}y_i}\prod_{i=1}^{d_1}\frac{\left(\gamma_{i}^{\boldY}\right)^{y_i}}{y_i!}.\label{eq:poisson:py|x:defn}  
\end{align}
Substituting the expressions for $P(\boldY|\msg)$ from~\eqref{eq:poisson:py|x:defn} and $\Tilde{Q}(\boldZQ|\boldYQ)$ into the term $\sum_{\mathcal{Y}_{\mathbf{z}}}P(\boldY|\msg) \tilde{Q}(\boldZQ|\boldYQ)$
\begin{align}
&\sum_{\mathcal{Y}_{\mathbf{z}}}P(\boldY|\msg)\Tilde{Q}(\boldZQ|\boldYQ) =\sum_{\mathcal{Y}_{\mathbf{z}}}\underbrace{{e^{-\msgsmall^n\sum_{i=1}^{d_1}\gamma_{i}^{\boldY}}}\msgsmall^{n\sum_{i=1}^{d_1}y_i}}_{=P(\boldY=\mathbf{y}|\msg=\msgsmall)}\times\nonumber\\
&\underbrace{ \prod_{i=1}^{d_1}\frac{\left(\gamma_{i}^{\boldY}\right)^{y_i}}{y_i!}}_{=P(\boldY=\mathbf{y}|\msg=\msgsmall)}\times\underbrace{\frac{\left(\sum_{i=1}^{d_1}y_i\right)!}{{\left(\sum_{i=1}^{d_1}y_i-\sum_{i=1}^{d_2}z_i\right)!\prod_{i=1}^{d_2}z_i!}}}_{=\Tilde{Q}(\boldZQ=\mathbf{z}|\boldYQ=\mathbf{y})}\times\nonumber\\ 
& \underbrace{\left(1\myminus\frac{\sum_{i=1}^{d_2}\gamma_{i}^{\boldZ}}{\sum_{i=1}^{d_1}\gamma_{i}^{\boldY}}\right)^{\sum_{i=1}^{d_1}y_i\myminus\sum_{i=1}^{d_2}z_i}\prod_{i=1}^{d_2}\left(\frac{{\gamma_{i}^{\boldZ}}}{\sum_{j=1}^{d_1}\gamma_{j}^{\boldY}}\right)^{z_i}}_{=\Tilde{Q}(\boldZQ=\mathbf{z}|\boldYQ=\mathbf{y})}
\intertext{Taking out the terms on the left hand side that do not depend on $\mathbf{y}$}
&\sum_{\mathcal{Y}_{\mathbf{z}}}P(\boldY|\msg)\Tilde{Q}(\boldZQ|\boldYQ)\myeq \boxed{e^{-\msgsmall^n\sum_{i=1}^{d_1}\gamma_{i}^{\boldY}} \prod_{i=1}^{d_2}\left(\gamma_{i}^{\boldZ}\right)^{z_i}\frac{1}{\prod_{i=1}^{d_2}z_i!}}\nonumber\\&\sum_{\mathcal{Y}_{\mathbf{z}}}\msgsmall^{n\sum_{i=1}^{d_1}y_i}\prod_{i=1}^{d_1}\frac{\left(\gamma_{i}^{\boldY}\right)^{y_i}}{y_i!}\left(1-\frac{\sum_{i=1}^{d_2}\gamma_{i}^{\boldZ}}{\sum_{i=1}^{d_1}\gamma_{i}^{\boldY}}\right)^{\sum_{i=1}^{d_1}y_i-\sum_{i=1}^{d_2}z_i}\nonumber\\
&\times\prod_{i=1}^{d_2}\left(\frac{1}{\sum_{j=1}^{d_1}\gamma_{j}^{\boldY}}\right)^{z_i}\times\frac{\left(\sum_{i=1}^{d_1}y_i\right)!}{\left(\sum_{i=1}^{d_1}y_i-\sum_{i=1}^{d_2}z_i\right)!}  
% \intertext{Rearranging the terms taken out of the summation}
% &\sum_{\mathcal{Y}_{\mathbf{z}}}P(\boldY=\mathbf{y}|\msg= \msgsmall)\Tilde{Q}(\boldZ=\mathbf{z}|\boldY=\mathbf{y}) =\nonumber\\
% &\left[\textcolor{blue}{e^{-\msgsmall^n\sum_{i=1}^{d_1}\gamma_{i}^{\boldY}}\times \prod_{i=1}^{d_2}\frac{\left(\gamma_{i}^{\boldZ}\right)^{z_i}}{z_i!}}\right]\times\sum_{\mathcal{Y}_{\mathbf{z}}}\left[\msgsmall^{n\sum_{i=1}^{d_1}y_i}\right.\nonumber\\
% &\times\prod_{i=1}^{d_1}\frac{\left(\gamma_{i}^{\boldY}\right)^{y_i}}{y_i!}\left(\sum_{i=1}^{d_1}y_i\right)!\left(1-\frac{\sum_{i=1}^{d_2}\gamma_{i}^{\boldZ}}{\sum_{i=1}^{d_1}\gamma_{i}^{\boldY}}\right)^{\sum_{i=1}^{d_1}y_i-\sum_{i=1}^{d_2}z_i}\nonumber\\
% &\times\left.\prod_{i=1}^{d_2}\left(\frac{1}{\sum_{j=1}^{d_1}\gamma_{j}^{\boldY}}\right)^{z_i}\times\frac{1}{\left(\sum_{i=1}^{d_1}y_i-\sum_{i=1}^{d_2}z_i\right)!} \right]
\end{align}
Decompose the term $\msgsmall^{n\sum_{i=1}^{d_1}y_i}$ as $\msgsmall^{n\sum_{i=1}^{d_1}y_i-n\sum_{i=1}^{d_2}z_i}$ multiplied by $\msgsmall^{n\sum_{i=1}^{d_2}z_i}$, and simplify the term $\prod_{i=1}^{d_2}\left(1/{\sum_{j=1}^{d_1}\gamma_{j}^{\boldY}}\right)^{z_i}$ as $1/{\left(\sum_{i=1}^{d_1}\gamma_{i}^{\boldY}\right)^{\sum_{i=1}^{d_2}z_i}}$. Then, combining the terms $\msgsmall^{n\sum_{i=1}^{d_1}y_i-n\sum_{i=1}^{d_2}z_i}$, $1/{\left(\sum_{i=1}^{d_1}\gamma_{i}^{\boldY}\right)^{\sum_{i=1}^{d_2}z_i}}$ and $\left(1-\frac{\sum_{i=1}^{d_2}\gamma_{i}^{\boldZ}}{\sum_{i=1}^{d_1}\gamma_{i}^{\boldY}}\right)^{\sum_{i=1}^{d_1}y_i-\sum_{i=1}^{d_2}z_i}$, we obtain:
\begin{align}
&\sum_{\mathcal{Y}_{\mathbf{z}}}P(\boldY|\msg)\Tilde{Q}(\boldZQ|\boldYQ)= e^{-\msgsmall^n\sum_{i=1}^{d_1}\gamma_{i}^{\boldY}}\boxed{\msgsmall^{n\sum_{i=1}^{d_2}z_i}}\ \times\nonumber\\
&\prod_{i=1}^{d_2}\frac{\left(\gamma_{i}^{\boldZ}\right)^{z_i}}{z_i!}\sum_{\mathcal{Y}_{\mathbf{z}}}\frac{\left(\sum_{i=1}^{d_1}y_i\right)!}{\left(\sum_{i=1}^{d_1}y_i-\sum_{i=1}^{d_2}z_i\right)!}\prod_{i=1}^{d_1}\frac{\left(\gamma_{i}^{\boldY}\right)^{y_i}}{y_i!}\times\nonumber\\
&\boxed{\frac{\left(\msgsmall^{n}\left(\sum_{i=1}^{d_1}\gamma_i^{\boldY}-\sum_{i=1}^{d_2}\gamma_{i}^{\boldZ}\right)\right)^{\sum_{i=1}^{d_1}y_i-\sum_{i=1}^{d_2}z_i}}{\left(\sum_{i=1}^{d_1}\gamma_{i}^{\boldY}\right)^{\sum_{i=1}^{d_1}y_i}}}.
\end{align}
Rearranging the terms $\prod_{i=1}^{d_1}\frac{(\gamma_i^{\boldY})^{y_i}}{y_i!}$, $\left(\sum_{i=1}^d y_i\right)!$, and $\frac{1}{(\sum_{i=1}^{d_1}\gamma_{i}^{\boldY})^{\sum_{i=1}^{d_1}y_i}}$ as follows:
\begin{align}
% \intertext{}
% &\sum_{\mathcal{Y}_{\mathbf{z}}}P(\boldY=\mathbf{y}|\msg= \msgsmall)\Tilde{Q}(\boldZ=\mathbf{z}|\boldY=\mathbf{y}) =\nonumber\\
% &\left[e^{-\msgsmall^n\sum_{i=1}^{d_1}\gamma_{i}^{\boldY}} \msgsmall^{n\sum_{i=1}^{d_2}z_i} \prod_{i=1}^{d_2}\frac{\left(\gamma_{i}^{\boldZ}\right)^{z_i}}{z_i!}\right]\sum_{\mathcal{Y}_{\mathbf{z}}}\left[\msgsmall^{n\sum_{i=1}^{d_1}y_i-n\sum_{i=1}^{d_2}z_i}\right.\nonumber\\
% &\times\left(\prod_{i=1}^{d_1}\left(\gamma_{i}^{\boldY}\right)^{y_i}\right){\frac{\left(\sum_{i=1}^{d_1}y_i\right)!}{\prod_{i=1}^{d_1}y_i!}}\frac{1}{\left(\sum_{i=1}^{d_1}y_i-\sum_{i=1}^{d_2}z_i\right)!}\nonumber\\
% &\times\left.\frac{\left(\sum_{i=1}^{d_1}\gamma_{i}^{\boldY}\myminus\sum_{i=1}^{d_2}\gamma_{i}^{\boldZ}\right)^{\sum_{i=1}^{d_1}y_i-\sum_{i=1}^{d_2}z_i}}{(\sum_{i=1}^{d_1}\gamma_{i}^{\boldY})^{\sum_{i=1}^{d_1}y_i}} \right]
&\sum_{\mathcal{Y}_{\mathbf{z}}}P(\boldY|\msg)\Tilde{Q}(\boldZQ|\boldYQ )\myeq e^{-\msgsmall^n\sum_{i=1}^{d_1}\gamma_{i}^{\boldY}}{\msgsmall^{n\sum_{i=1}^{d_2}z_i}}\prod_{i=1}^{d_2}\frac{\left(\gamma_{i}^{\boldZ}\right)^{z_i}}{z_i!}\nonumber\\
&\sum_{\mathcal{Y}_{\mathbf{z}}}\frac{1}{\left(\sum_{i=1}^{d_1}y_i-\sum_{i=1}^{d_2}z_i\right)!}\boxed{\frac{\left(\sum_{i=1}^{d_1}y_i\right)!}{\prod_{i=1}^{d_1}y_i!}{\prod_{i=1}^{d_1}\left(\frac{\gamma_{i}^{\boldY}}{\sum_{j=1}^{d_1}\gamma_{j}^{\boldY}}\right)^{y_i}}}\nonumber\\
&\left(\msgsmall^n\left(\sum_{i=1}^{d_1}\gamma_{i}^{\boldY}\myminus\sum_{i=1}^{d_2}\gamma_{i}^{\boldZ}\right)\right)^{\sum_{i=1}^{d_1}y_i-\sum_{i=1}^{d_2}z_i}
\end{align}
Define $\mathcal{Y}^{k}=\left\{\mathbf{y}\in(\Nnot)^{d_1} \st \sum_{i=1}^{d_1}y_i=k\right\}$, i.e. the set of all $d_1$-dimensional count vectors such that they sum up to $k$. {Since, $\mathcal{Y}_{\mathbf{z}}$ is the collection of all $d_1$-dimensional count vectors such that their sum is greater than $\sum_{i=1}^{d_1}z_i$, we can express $\mathcal{Y}_{\mathbf{z}}$ as an union of $\mathcal{Y}^k$. More concretely,} $\mathcal{Y}_{\mathbf{z}}=\cup_{k=\sum_{i=1}^{d_2}z_i}^{\infty}\mathcal{Y}^k$. Furthermore, we also have that $\mathcal{Y}^{k_1}$ and $\mathcal{Y}^{k_2}$ are disjoint sets iff $k_1\neq k_2$. Using these previous two facts, we can decompose the summation $\sum_{\mathcal{Y}_{\mathbf{z}}}$ as a double summation, i.e. $\sum_{\mathcal{Y}_{\mathbf{z}}}\equiv \sum_{k=\sum_{i=1}^{d_2}z_i}^{\infty}\sum_{\mathcal{Y}^k}$. Substituting this fact in the previous equation, we get:
\begin{align}
&\sum_{\mathcal{Y}_{\mathbf{z}}}P(\boldY|\msg)\Tilde{Q}(\boldZ|\boldY) =e^{-\msgsmall^n\sum_{i=1}^{d_1}\gamma_{i}^{\boldY}} \msgsmall^{n\sum_{i=1}^{d_2}z_i} \prod_{i=1}^{d_2}\frac{\left(\gamma_{i}^{\boldZ}\right)^{z_i}}{z_i!}\nonumber\\
&\times\sum_{k=\sum_{i=1}^{d_2}z_i}^{\infty}\sum_{\mathcal{Y}^k}\underbrace{\frac{\left(\sum_{i=1}^{d_1}y_i\right)!}{\prod_{i=1}^{d_1}y_i!}\prod_{i=1}^{d_1}\left(\frac{\gamma_{i}^{\boldY}}{\sum_{j=1}^{d_1}\gamma_{j}^{\boldY}}\right)^{y_i}}_{=\text{Term}\ 1}\times\nonumber\\
&\underbrace{{\frac{\left(\msgsmall^n\left(\sum_{i=1}^{d_1}\gamma_{i}^{\boldY}\myminus\sum_{i=1}^{d_2}\gamma_{i}^{\boldZ}\right)\right)^{\sum_{i=1}^{d_1}y_i-\sum_{i=1}^{d_2}z_i}}{\left(\sum_{i=1}^{d_1}y_i-\sum_{i=1}^{d_2}z_i\right)!}}}_{\text{Term} 2}
\end{align}
Note that Term 2 in the above equation only contains terms that depend on $\sum_{i=1}^{d_1}y_i$. Since all the $\mathbf{y}$ present in $\mathcal{Y}^k$ sum up to the same value, i.e. $k$, Term 2 is a constant with respect to the inner summation $(\sum_{\mathcal{Y}^k})$ in the above equation. Hence, moving Term 2 out of the inner summation and replacing all $\sum_{i=1}^{d_1}y_i$ by $k$ in Term 2, we obtain:
\begin{align}
&\sum_{\mathcal{Y}_{\mathbf{z}}}P(\boldY|\msg)\Tilde{Q}(\boldZ|\boldY) = e^{-\msgsmall^n\sum_{i=1}^{d_1}\gamma_{i}^{\boldY}} \msgsmall^{n\sum_{i=1}^{d_2}z_i} \prod_{i=1}^{d_2}\frac{\left(\gamma_{i}^{\boldZ}\right)^{z_i}}{z_i!}\nonumber\\
&\times\sum_{k=\sum_{i=1}^{d_2}z_i}^{\infty}\underbrace{{\frac{\left(\msgsmall^n\left(\sum_{i=1}^{d_1}\gamma_{i}^{\boldY}\myminus\sum_{i=1}^{d_2}\gamma_{i}^{\boldZ}\right)\right)^{k-\sum_{i=1}^{d_2}z_i}}{\left(k-\sum_{i=1}^{d_2}z_i\right)!}}}_{\text{Term} 2}\nonumber\\
&\times\left.\sum_{\mathcal{Y}^k}\underbrace{\frac{\left(\sum_{i=1}^{d_1}y_i\right)!}{\prod_{i=1}^{d_1}y_i!} \prod_{i=1}^{d_1}\left(\frac{\gamma_{i}^{\boldY}}{\sum_{j=1}^{d_1}\gamma_{j}^{\boldY}}\right)^{y_i}}_{=\text{Term 1}}
\right]\label{eq:poisson:lemma2:31}
\end{align}
Note that Term 1 is just a Multinomial distribution, with $n=\sum_{i=1}^{d_1} y_i$, $\mathbf{k}=\mathbf{y}$, and $\mathbf{p}=\begin{bmatrix}
    \frac{\gamma_{11}}{\sum_{i=1}^{d_1}\gamma_{i}^{\boldY}}&\cdots&    \frac{\gamma_{1d_1}}{\sum_{i=1}^{d_1}\gamma_{i}^{\boldY}}
\end{bmatrix}$. Since $\mathcal{Y}^k$ represents the whole support of this Multinomial distribution, the inner summation is just summing up a Multinomial distribution over all its support, i.e. Term 1 is equal to $1$. Substituting this fact in the above equation:
% &\sum_{\mathcal{Y}^k}\frac{\left(\sum_{i=1}^{d_1}y_i\right)!}{\prod_{i=1}^{d_1}y_i!} \prod_{i=1}^{d_1}\left(\frac{\gamma_{i}^{\boldY}}{\sum_{j=1}^{d_1}\gamma_{j}^{\boldY}}\right)^{y_i} = 1\label{eq:poisson:lemma2:multinomial-Y}
% \intertext{Combining~\eqref{eq:poisson:lemma2:31} and~\eqref{eq:poisson:lemma2:multinomial-Y}, we get:}
\begin{align}
&\sum_{\mathcal{Y}_{\mathbf{z}}}P(\boldY|\msg)\Tilde{Q}(\boldZ|\boldY) = e^{-\msgsmall^n\sum_{i=1}^{d_1}\gamma_{i}^{\boldY}} \msgsmall^{n\sum_{i=1}^{d_2}z_i} \prod_{i=1}^{d_2}\frac{\left(\gamma_{i}^{\boldZ}\right)^{z_i}}{z_i!}\nonumber\\
&\times\sum_{k'=0}^{\infty}\frac{\left(\msgsmall^n\left(\sum_{i=1}^{d_1}\gamma_{i}^{\boldY}\myminus\sum_{i=1}^{d_2}\gamma_{i}^{\boldZ}\right)\right)^{k'}}{k'!},
\end{align}
where $k'=k-\sum_{i=1}^{d_2}z_i$. Using the fact that the Taylor Series of $e^\msgsmall=\sum_{k=0}^{\infty}\frac{\msgsmall^k}{k!}$ in the above equation, we obtain:
\begin{align}
&\sum_{\mathcal{Y}_{\mathbf{z}}}P(\boldY|\msg)\Tilde{Q}(\boldZ|\boldY) = \left[e^{-\msgsmall^n\sum_{i=1}^{d_1}\gamma_{i}^{\boldY}} \msgsmall^{n\sum_{i=1}^{d_2}z_i} \right.\nonumber\\
&\times\left.\prod_{i=1}^{d_2}\frac{\left(\gamma_{i}^{\boldZ}\right)^{z_i}}{z_i!}\right]\times{e^{\msgsmall^n\sum_{i=1}^{d_1}\gamma_{i}^{\boldY}-\msgsmall^n\sum_{i=1}^{d_2}\gamma_{i}^{\boldZ}}}
\intertext{Combining the terms $e^{\msgsmall^n\sum_{i=1}^{d_1}\gamma_{i}^{\boldY}-\msgsmall^n\sum_{i=1}^{d_2}\gamma_{i}^{\boldZ}}$ and $e^{-\msgsmall^n\sum_{i=1}^{d_1}\gamma_{i}^{\boldY}}$, we obtain:}
&\sum_{\mathcal{Y}_{\mathbf{z}}}P(\boldY|\msg)\Tilde{Q}(\boldZ|\boldY) = e^{-\msgsmall^n\sum_{i=1}^{d_2}\gamma_{i}^{\boldZ}} \msgsmall^{n\sum_{i=1}^{d_2}z_i} \prod_{i=1}^{d_2}\frac{\left(\gamma_{i}^{\boldZ}\right)^{z_i}}{z_i!}\label{eq:poisson:lemma2:finaleq}
\intertext{Since $P(\boldZ|\msg)$ is also a product of Poisson distributions, i.e.:}
&P(\boldZ=\mathbf{z}|\msg=\msgsmall)=\prod_{i=1}^{d_2}\Po\left(z_i;\gamma_{i}^{\boldZ}\msgsmall^n\right)\nonumber\\ 
&= e^{-\msgsmall^n\sum_{i=1}^{d_2}\gamma_{i}^{\boldZ}}\msgsmall^{n\sum_{i=1}^{d_2}z_i}\prod_{i=1}^{d_2}\frac{\left(\gamma_{i}^{\boldZ}\right)^{z_i}}{z_i!},\label{eq:poisson:pz|x:defn}
\intertext{Combining~\eqref{eq:poisson:pz|x:defn},~\eqref{eq:poisson:lemma2:equivalence} and~\eqref{eq:poisson:lemma2:finaleq}, we obtain:}
&\Tilde{Q}(\boldZ|M)=\sum_{\mathbf{y}}P(\boldY|\msg)\Tilde{Q}(\boldZ|\boldY) = P(\boldZ|\msg)\label{eq:poisson:lemma2:finaleq2}\\
&\Rightarrow\Tilde{Q}(\msg,\boldZ)=P(\msg)P(\boldZ|\msg)=P(\msg,\boldZ)\label{eq:poisson:lemma2:finaleq3}
\end{align}
From~\eqref{eq:poisson:lemma2:finaleq4} and~\eqref{eq:poisson:lemma2:finaleq3}, we can conclude that $\Tilde{Q}(\msg,\boldY,\boldZ)$ has the same marginals $P(\msg,\boldY)$ and $P(\msg,\boldZ)$ as $P(\msg,\boldY,\boldZ)$, and hence lies in the set $\Delta_P$.
\end{proof}
\subsection{Proof of Lemma~\ref{lemma:multinom_conditional}}\label{appx:proof-lemma-multinomial}
\begin{proof}
    Without loss of generality, reshuffle the class indices of $\boldz$ such that $\idxset \myeq \{ i \st 1 \leq i \leq w \}$, where $|\idxset| \myeq w$. Let $\sumset \myeq \{ \boldy' \in \mathbb{N}_0^{\ky-1} \st  \sum_{i=1}^{\ky-1} y_i \leq n \}$, where $\mathbb{N}_0^{\ky-1}$ is a $(\ky-1)$-dimensional vector. By the law of total probability, 
    \begin{align}
        P(Z = \boldz) &= \sum_{\boldy \in C} P(Y = \boldy) P(Z = \boldz | Y = \boldy)
        \\
        &= \sum_{\boldy \in C} \frac{n!}{\prod_{i=1}^{\ky} y_i!} \prod_{i=1}^{\ky}p_i^{y_i} \cdot \frac{\left( \sum_{j \in \idxset} y_{j}\right) !}{\prod_{i=1}^{\kz}z_i!} \prod_{i=1}^{\kz} q_i^{z_i}.
    \end{align}
    Remove from the sum all terms that do not depend on $\mathbf{y}$. Note that $z_{\kz} \myeq \sum_{j \in \idxset} y_j - \sum_{i=1}^{\kz - 1} z_i$ and $q_{\kz}^{z_{\kz}} \myeq \left(1- \sum_{i=1}^{\kz} q_i \right)^{z_{\kz}}$ are both functions of $\mathbf{y}$. Let $\sumoverz \myeq \mzsum$ and $\sumoverq \myeq \sum_{i=1}^{\kz-1} q_i$.
    \begin{align}
        &P(Z = \boldz) = 
        \left( \frac{n!}{\prod_{i=1}^{\kz-1} z_i!} \prod_{i=1}^{\kz-1} q_i^{z_i} \right) 
        \nonumber
        \\
        &\times \sum_{\boldy \in \sumset} 
        \frac{\prod_{i=1}^{\ky} p_i^{y_i} \cdot \left( \sum_{j \in \idxset} y_{j} \right)!}{\prod_{i=1}^{\ky} y_i! \left( \sum_{j \in \idxset} y_{j} - \sumoverz \right)!}
        \left( 1- \sumoverq\right)^{\sum_{j \in \idxset} y_{j} - \sumoverz}.
    \end{align}
    Form a multinomial coefficient outside the sum:
    \begin{align}
        \label{eq:multinom_z_all}
        &P(Z = \mathbf{z}) = 
        \left( \frac{n!}{\prod_{i=1}^{\kz-1} z_i! \cdot \left( n - \sumoverz \right)!} \prod_{i=1}^{\kz-1} q_i^{z_i} \right)
        \nonumber
        \\
        &\times \sum_{\mathbf{y} \in \sumset} 
        \frac{\left( n - \sumoverz \right)! \prod_{i=1}^{\ky} p_i^{y_i} \cdot \left( \sum_{j \in \idxset} y_{j} \right)!}{\prod_{i=1}^{\ky} y_i! \left( \sum_{j \in \idxset} y_{j} - \sumoverz \right)!} \left( 1- \sumoverq \right)^{\sum_{j \in \idxset} y_{j} - \sumoverz}.
    \end{align}
    Let $\idxsetc = \{ i \st w < i \leq \ky \}$ be the set of indices that are not taken in the number of trials for $Z$. Consider the sum of $\mathbf{y} \in C$ alone. Separating terms that belong to $\idxset$ and $\idxsetc$,
    \begin{align}
        &\sum_{\mathbf{y} \in C} \frac{\left( n - \sumoverz \right)!}{\prod_{j \in \idxsetc} y_j! \left( \sum_{j \in \idxset} y_j - \sumoverz \right)!} \prod_{j \in \idxsetc} p_j^{y_j}
        \nonumber
        \\
        &\times \left( 1- \sumoverq \right)^{\sum_{j \in \idxset} y_{j} - \sumoverz}
        \cdot
        \frac{\left( \sum_{j \in \idxset} y_j \right)!}{\prod_{j \in \idxset} y_j!} \prod_{j \in \idxset} p_j^{y_j} \nonumber
    \end{align}
    Let $\sumset' = \{ \boldy' \in \mathbb{N}_0^{\ky - 1} \st \sum_{i=1}^{\ky - 1} y_i - \sum_{i=1}^{\kz-1} z_i \leq n - \sum_{i=1}^{\kz - 1} z_i \}$. Note $\mathbf{y} \in C \Leftrightarrow \mathbf{y} \in C'$, so we can equivalently sum over elements in $C'$. Perform a change of variable with $u = \sum_{j \in \idxset} y_j - \sumoverz$ and define $B = \{ y_j \in \mathbf{y} \st j \in \idxset, \sum_{j \in \idxset} y_j = u + \sumoverz \}$. Then the sum becomes
    \begin{align}
        &= \sum_{\mathbf{y} \in \sumset'} \frac{\left( n - \sumoverz \right)!}{\prod_{j \in \idxsetc} y_j! u!} \prod_{j \in \idxsetc} p_j^{y_j} \left( 1- \sumoverq  \right)^u 
        \nonumber
        \\
        &\times \sum_{y_j \in B} 
        {u + \sumoverz \choose y_{j_1}, ..., y_{j_m}} \prod_{j \in \idxset} p_j^{y_j}.
        \\
        \intertext{Note that the inner sum simplifies by the multinomial theorem, part of which we remove from the sum:}
        &= \sum_{\mathbf{y} \in \sumset'} \frac{\left( n - \sumoverz \right)!}{\prod_{j \in \idxsetc} y_j! u!} \prod_{j \in \idxsetc} p_j^{y_j} \left( 1- \sumoverq  \right)^u \left( \sum_{j \in \idxset} p_j \right)^{u + \sumoverz}
        \\
        &= \left( \sum_{j \in \idxset} p_j \right)^{\sumoverz} \sum_{\mathbf{y} \in C'} \frac{\left( n - \sumoverz \right)!}{\prod_{j \in \idxsetc} y_j! \cdot u!} \prod_{j \in \idxsetc} p_j^{y_j} 
        \nonumber
        \\
        &\times \left( \sum_{j \in \idxset} p_j - \sumoverq \sum_{j \in \idxset} p_j \right)^u.
        \\
        \intertext{Reapply the multinomial theorem to the sum of $\mathbf{y}$ over $\sumset'$:}
        &= \left( \sum_{j \in \idxset} p_j \right)^{\sumoverz} \left( \sum_{j \in \idxsetc} p_j + \sum_{j \in \idxset} p_j - \sumoverq \sum_{j \in \idxset} p_j \right)^{n - \sumoverz}
        \\
        &= \left( \sum_{j \in \idxset} p_j \right)^{\sumoverz} 
        \left( 1 - \sumoverq \sum_{j \in \idxset} p_j \right)^{n - \sumoverz}.
    \end{align}
    Returning to the expression in~\eqref{eq:multinom_z_all}, the marginal becomes
    \begin{align}
        &P(Z = \mathbf{z}) = \frac{n! \prod_{i=1}^{\kz-1} q_i^{z_i} }{\prod_{i=1}^{\kz-1} z_i! \cdot \left( n - \sumoverz \right)!} 
        \nonumber
        \\
        &\times \left( \sum_{j \in \idxset} p_j \right)^{\sumoverz} \left( 1 - \sumoverq \sum_{j \in \idxset} p_j \right)^{n - \sumoverz}
        \\
        &= \frac{n! \prod_{i=1}^{\kz-1} \left( \sum_{j \in \idxset} p_j \cdot q_i\right)^{z_i}}{\prod_{i=1}^{\kz-1} z_i! \cdot \left( n - \sumoverz \right)!} \left( 1 - \sumoverq \sum_{j \in \idxset} p_j \right)^{n - \sumoverz}
        \\
        &= \text{Multinomial} \left( \mathbf{z}; n, \mathbf{q}^* \right),
    \end{align}
    where $\mathbf{q}^*$ is defined as in~\eqref{eq:multinomial_distribution}.
\end{proof}
\section{}\label{sec:appx:B}
\subsubsection*{Deriving the p.m.f. of the Multivariate Poisson Distribution}
Let $\boldK\sim\Po(d,d',\boldsymbol{\Lambda})$, where $\boldK$ is a $d$-dimensional random vector. We know that $\boldK=A\boldK_g$, where $A=[A_1\hdots A_{d'}]$. Let the dimension of $\boldK^g$ be $d_{\boldK^g}$ and $\mathcal{S}_{\mathbf{k}}=\{\mathbf{k}^g\in\Nnot^{d_{\boldK^g}} \st \mathbf{k}=A\mathbf{k}^g\}$. Then using the fact $\boldK=A\boldK^g$, the p.m.f. of $\boldK$ can be expressed using the p.m.f. of $\boldK^g$ in the following manner:
\begin{align}
P(\boldK=\mathbf{k})=\sum_{\mathbf{k}^g\in\mathcal{S}_{\mathbf{k}}}P(\boldK^g=\mathbf{k}^g)\label{eq:pmf:1}
\end{align}
Since $\boldK^g$ is just a collection of mutually independent Poisson random variables, we can write $P(\boldK^g=\mathbf{k}^g)$ as a product of scalar Poisson distributions, i.e.
\begin{align}
&P(\boldK^g=\mathbf{k}^g)=\prod_{j=1}^{d'}\prod_{(i_1,\hdots,i_j)\in\mathbb{A}_j^d}e^{\lambda_{i_1,\hdots, i_{j}}}\frac{\lambda_{i_1,\hdots, i_{j}}^{k_{i_1,\hdots, i_j}^g}}{k_{i_1,\hdots, i_j}^g!}.\label{eq:pmf:2}
\end{align}
Substituting~\eqref{eq:pmf:2} into~\eqref{eq:pmf:1}, and collecting all $e^{\lambda_{i_1,\hdots, i_j}}$ terms we obtain:
\begin{align}
% P(\boldK\myeq\mathbf{k})&\myeq\sum_{\mathbf{k}^g\in\mathcal{S}_{\mathbf{k}}}\prod_{j=1}^{d'}\prod_{(i_1,\hdots,i_j)\in\mathbb{A}_j^d}\Po(\lambda_{i_1\hdots i_{j}})\label{eq:pmf:3}\\
% % \prod_{(i_1,i_2)\in\mathbb{A}_2^d}\Po(\lambda_{i_1i_2})
% % \right.\nonumber\\
% % &\left.\times\hdots\times\prod_{(i_1,\hdots,i_{d'})\in\mathbb{A}_{d'}^d}\Po(\lambda_{i_1\hdots i_{d'}})\right]
&P(\boldK\myeq\mathbf{k}){\myeq} e^{-\mathbf{1}^T\boldsymbol{\Lambda}}\sum_{\mathbf{k}^g\in\mathcal{S}_{\mathbf{k}}}\prod_{j=1}^{d'}\prod_{(i_1,\hdots,i_j)\in\mathbb{A}_j^d}\frac{\lambda_{i_1,\hdots, i_{j}}^{k_{i_1,\hdots, i_j}^g}}{k_{i_1,\hdots, i_j}^g!}\label{eq:pmf:4}
\end{align}
 Decompose the matrix $A=[A_1\ A']$, where $A'=[A_2\hdots A_{d'}]$ and $\mathbf{k}^g=[(\mathbf{k}_1^g)^T\ \mathbf{k}'^T]^T$, where $\mathbf{k}_1^g=[k^g_1\hdots k^g_{d}]^T$ and $\mathbf{k}'$ contains the rest of the elements in $\mathbf{k}^g$. Using the fact that $\mathbf{k}^g\in\mathcal{S}_{\mathbf{k}}\Rightarrow \mathbf{k}=A\mathbf{k}^g$ we have: $\mathbf{k}=A_1\mathbf{k}_1^g+A'\mathbf{k}'$, which combining with the fact that $A_1$ is an identity matrix, we obtain:
 \begin{align}
\mathbf{k}=\mathbf{k}_1^g+A'\mathbf{k}'\Rightarrow \mathbf{k}_1^g = \mathbf{k}-A'\mathbf{k'}\Rightarrow k_i^g = k_i-(a_i')^T\mathbf{k'},\label{eq:pmf:5}
 \end{align}
 where $a_i'$ is the $i$-th row of $A'$ and $i\in[d]$. Substituting~\eqref{eq:pmf:5} into~\eqref{eq:pmf:4}:
 \begin{align}
P(\boldK\myeq\mathbf{k})\myeq &e^{-\mathbf{1}^T\boldsymbol{\Lambda}}\sum_{\mathbf{k}'\in\mathcal{S}_{\mathbf{k}'}}\left[\prod_{i_1\in\mathbb{A}_1^d}\frac{\lambda_{i_1}^{k_{i}-(a_i')^T\mathbf{k'}}}{(k_{i}-(a_i')^T\mathbf{k'})!}\times
\right.\nonumber\\
&\left.\prod_{j=2}^{d'}\prod_{(i_1,\hdots,i_j)\in\mathbb{A}_j^d}\frac{\lambda_{i_1,\hdots, i_{j}}^{k_{i_1,\hdots, i_j}^g}}{k_{i_1,\hdots, i_j}^g!}\right]\label{eq:pmf:6}
\end{align}
 where the summation constraint is transformed to $\mathcal{S}_{\mathbf{k}'}=\{\mathbf{k}'\in\Nnot^{l-d} \st (a_i')^T\mathbf{k'}\leq k_i\ \forall\ i\in[d]\}$. The equivalence of the summations over the sets $\mathcal{S}_{\mathbf{k}'}$ and $\mathcal{S}_{\mathbf{k}}$ can be derived by considering that $k_i^g=k_i-(a_i')^T\mathbf{k}'$, and $k_i^g\geq 0$. Collecting all the factorial terms, we get: 
 \begin{align}
  &Q(\mathbf{k},\mathbf{k}')=\prod_{i=1}^{d}\frac{1}{(k_i-a_i'^T\mathbf{k}')!}\prod_{j=2}^{d'}\prod_{(i_1,\hdots,i_j)\in\mathbb{A}_j^{d}}\frac{1}{k^g_{i_1,\hdots ,i_j}!}, \label{eq:pmf:7} 
\end{align}
Substituting~\eqref{eq:pmf:7} into~\eqref{eq:pmf:6}, we obtain:
 \begin{align}
P(\boldK\myeq\mathbf{k})\myeq& e^{-\mathbf{1}^T\boldsymbol{\Lambda}}\sum_{\mathbf{k}'\in\mathcal{S}_{\mathbf{k}'}}\left[\prod_{i_1\in\mathbb{A}_1^d}\lambda_{i_1}^{k_{i}-(a_i')^T\mathbf{k'}}
\right.\nonumber\\
&\left.\prod_{j=2}^{d'}\prod_{(i_1,\hdots,i_j)\in\mathbb{A}_j^d}\lambda_{i_1,\hdots, i_j}^{k_{i_1,\hdots, i_j}^g}\right]Q(\mathbf{k},\mathbf{k}')\label{eq:pmf:8}
\end{align}
Taking the term $\prod_{i_1\in\mathbb{A}_1}\lambda_{i_1}^{k_i}$ out of the summation as it does not depend upon $\mathbf{k}'$, we obtain:
 \begin{align}
P(\boldK\myeq\mathbf{k})\myeq& e^{-\mathbf{1}^T\boldsymbol{\Lambda}}\prod_{i_1\in\mathbb{A}_1^d}\lambda_{i_1}^{k_i}\sum_{\mathbf{k}'\in\mathcal{S}_{\mathbf{k}'}}\left[\prod_{i_1\in\mathbb{A}_1^d}\lambda_{i_1}^{-(a_i')^T\mathbf{k'}}
\right.\nonumber\\&\left.\prod_{j=2}^{d'}\prod_{(i_1,\hdots,i_j)\in\mathbb{A}_j^d}\lambda_{i_1,\hdots, i_j}^{k_{i_1,\hdots, i_j}^g}\right]Q(\mathbf{k},\mathbf{k}')\label{eq:pmf:9}
\end{align}
Notice that $(a_{i}')^T\mathbf{k}'$ contains all the elements of $\mathbf{k}'$ which contains $i$ in their subscript. Hence, expanding the term $\prod_{i_1\in\mathbb{A}_1^d}\lambda_{i_1}^{-(a_i')^T\mathbf{k'}}$ and distributing over the product, we obtain the desired form:
\begin{align}
 &P(\boldK\myeq\mathbf{k})\myeq e^{-\mathbf{1}^T\boldsymbol{\Lambda}}\prod_{i_1\in\mathbb{A}_1^d}\lambda_{i_1}^{k_i}\sum_{\mathbf{k}'\in\mathcal{S}_{\mathbf{k}'}}\left[\prod_{(i_1,i_2)\in\mathbb{A}_2^d}\left(\frac{\lambda_{i_1,i_2}}{\lambda_{i_1}\lambda_{i_2}}\right)^{k_{i_1,i_2}^g}
\right.\nonumber\\
&\left.\times\hdots\times\prod_{(i_1,\hdots,i_{d'})\in\mathbb{A}_{d'}^d}\left(\frac{\lambda_{i_1,\hdots, i_{d'}}}{\lambda_{i_1}\hdots\lambda_{i_{d'}}}\right)^{k_{i_1,\hdots, i_{d'}}}\right]Q(\mathbf{k},\mathbf{k}').\label{eq:pmf:10}   
\end{align}

\end{document}